\documentclass[12pt]{article}
\usepackage[margin=1.25in]{geometry}
\usepackage{mathrsfs}
\usepackage[title,toc,titletoc]{appendix}
\usepackage{titlesec}
\titleformat{\section}[block]{\large\normalfont\filcenter}{\thesection.}{.5em}{}
\titleformat{\subsection}[hang]{\itshape\bfseries}{\thesubsection.}{.5em}{}
\usepackage[dvipsnames]{xcolor}
\usepackage{bbm}
\usepackage{bm}
\usepackage{amsmath, amssymb}	
\usepackage[shortlabels]{enumitem}
\usepackage[hang,flushmargin]{footmisc}
\usepackage[round]{natbib}
\usepackage[colorlinks=true,linkcolor = blue, citecolor=blue, hyperfootnotes=false, hyperindex,breaklinks]{hyperref}
\usepackage{cleveref}
\usepackage{bigints}

\usepackage{pgf}
\usepackage[justification=centering]{caption}
\usepackage{subcaption}
\usepackage{tikz}
\usetikzlibrary{patterns, decorations.pathreplacing}
\usepackage{hhline}
\usepackage{float}
\usepackage{siunitx}
\usepackage{textcomp}
\usepackage[subtle]{savetrees}
\usepackage{graphics}
\usepackage{graphicx}

\newtheorem{theorem}{Theorem}

\newtheorem{remark}{Remark}
\newtheorem{corollary}{Corollary}
\newtheorem{definition}{Definition}
\newtheorem{lemma}{Lemma}
\newtheorem{assumption}{Assumption}
\newenvironment{proof}[1][Proof]{\textbf{#1.} }{\  \rule{0.5em}{0.5em}}
\DeclareMathOperator*{\argmax}{arg\,max}

\DeclareMathOperator*{\supp}{supp}

\usepackage{setspace}
\usepackage{multirow}
\definecolor{shade}{gray}{.7}
\begin{document}
\title{Persuaded Search\thanks{For helpful comments and discussions, we thank Nageeb Ali, Nima Haghpanah, Emir Kamenica, Erik Madsen, Fedor Sandomirskiy, Andy Skrzypacz, Juuso Toikka, Mark Whitmeyer, Kun Zhang, and seminar audiences at Brown,  Duke, NYU Stern, Princeton, EC24 Workshop on
Information Acquisition, Conference on Mechanism and Institution Design, Cowles Conference on Economic Theory, SITE, and Stony Brook International Conference on Game Theory. This paper supersedes a previous draft entitled ``Certification in Search Markets.''}}
\author{Teddy Mekonnen\thanks{Department of Economics, Brown University. Contact: \href{mailto:mekonnen@brown.edu}{mekonnen@brown.edu}}\hspace*{1.25em} Zeky Murra-Anton\thanks{Market Development, ISO New England. Contact: \href{mailto:zmurraanton@iso-ne.com }{zmurraanton@iso-ne.com}. This paper has not been funded or endorsed by ISO New England nor does it represent ISO New England's views.} \hspace*{1.25em} Bobak Pakzad-Hurson\thanks{Department of Economics, Brown University. Contact: \href{mailto:bph@brown.edu}{bph@brown.edu}}}
\date{\today}

\maketitle
\thispagestyle{empty}
\setcounter{page}{0}
\begin{abstract}
We consider sequential search by an agent who cannot observe the quality of goods but can acquire information from a profit-maximizing principal with limited commitment power. The principal can charge higher prices for more informative signals, but high future prices discourage continued search, thereby reducing the principal's profits. A unique stationary equilibrium outcome exists: the principal $(i)$ induces the socially efficient stopping rule, $(ii)$ extracts the full surplus, and $(iii)$ persuades the agent against settling for marginal goods, extending the duration of surplus extraction. However, introducing an additional, free source of information can lead to inefficiency in equilibrium.\end{abstract}

\newpage

\section{Introduction}
 \begingroup
\allowdisplaybreaks

Agents in search and matching markets often contend with frictions stemming from incomplete information. An agent who does not perfectly observe the quality of the goods he samples during his search may mistakenly select a good of inferior quality or may abandon his search altogether. Given these inefficiencies, it is unsurprising that searching agents turn to information brokers so as to reduce the inherent uncertainty in various markets: consumers seek information on products before making purchases, home buyers commission geological risk assessments for natural disasters before closing, and individuals turn to advice from matchmakers to identify compatible spouses.\footnote{Each author of this paper has solicited CarFax reports during searches for used automobiles. Two authors have commissioned geological reports while searching for houses. None of the authors admits to having contacted a matchmaker, but at least one probably should.}

We refer to the interaction between an ``agent" (he) who searches for a ``good'' and a ``principal" (she) who brokers information as a \emph{persuaded search} game. Because the objectives of the agent and principal differ---the agent seeks to maximize the outcome of his search while the principal seeks to maximize the payments she receives---natural questions arise:  How much information about goods' quality should the principal optimally convey? How much should she charge?  

The answers to these questions potentially depend on salient features of the search market. To fix ideas, consider the market for pre-employment skills testing. A firm (``agent'') seeks to fill a job vacancy, but does not observe the productivity of each candidate employee (``goods''). However, the firm can contract with an information broker (``principal'') who adjudicates candidate productivity through skill assessments and job simulations. Such pre-employment testing services are utilized by 82\% of US private-sector firms,\footnote{https://www.sparcstart.com/wp-content/uploads/2018/03/2017-CandE-Report.pdf.}  giving rise to an industry with 2 billion USD in annual revenue.\footnote{https://www.polarismarketresearch.com/industry-analysis/candidate-skills-assessment-market. The annual revenue figure is given for the most recent year reported, 2021, and  is projected to nearly triple by 2030.} Evidence shows that pre-employment testing  dramatically increases the quality of worker hired, and provides significant information to the firm that is otherwise unavailable at the time of hiring \citep{aut08, hof18}. Typically, a firm contracts with the information broker on a short-term basis with payments on a per-test basis,\footnote{For example, ProfileXT charges per candidate employee assessed. This payment structure is possibly due to difficulties committing to future assessments in the presence of changing skill demands over time.} there is a single information broker who contracts with the firm,\footnote{\label{footnote_diamond}A single information broker could be due to the existence of a ``natural monopoly'' in certain industries. For example, nurses must pass an exam to be certified in Advanced Cardiovascular Life Support (a requirement for certain positions), but the American Heart Association is the only accredited body to certify such tests. De facto monopoly power  can be rationalized in a model with multiple information brokers if searching firms must additionally engage in costly search for new information brokers  \'a la \cite{dia71}.} and the information broker is the sole source of information for the searching firm.\footnote{In low-skill service jobs, for example, the candidate assessment plays a crucial role in the hiring process, as it typically occurs before the job interview (see https://www.sparcstart.com/wp-content/uploads/2018/03/2017-CandE-Report.pdf) and the test results often serve as the deciding factor in hiring decisions. \cite{aut08} suggest that pre-employment test results affect only hiring and not wages in the particular market they study. More broadly, test results are typically not revealed to candidate employees and are therefore unlikely to factor into wage setting.}
 
Of course, some search markets differ along various dimensions---for example, real estate agents and headhunters are often paid only upon the successful termination of search, agents in some markets may have the ability to contract with multiple competing principals, and agents in some markets may have external information about candidate quality. We proceed by first studying the equilibrium of a model reflecting the pre-employment testing industry, and then discuss the role of institutional features in guiding the equilibrium outcome. 

To that end, we present a parsimonious model that embeds an information design framework \`a la \cite{kam11} into a simple sequential search problem \`a la \cite{mcc70}.\footnote{\label{footnote_mccall}We note that McCall motivates his model by considering a worker searching for a firm, while our motivating example considers a firm searching for a worker. This difference is merely expositional; \cite[p. 113]{mcc70} notes, ``The analysis presented here is directed to the employee's job-searching strategy. Obviously, similar methods could be applied to investigate the employer's job market behavior.''} A long-lived agent  seeks to fill a single vacancy by randomly sampling from a continuum of goods that are differentiated in their quality. The agent samples one good per period and exits the market upon making a selection--he prefers to match with a higher quality good but is also impatient and prefers to match sooner rather than later. We depart from the classic sequential search literature by assuming that sampling a good reveals no information on the quality of that good to the agent. Instead, the agent can acquire information from a long-lived  principal. Similarly to \cite{ber18}, and unlike much of the information design literature, we assume that the principal has no intrinsic preferences over the outcome of the agent's search. Instead, she sells information to the agent and seeks to maximize her profits. In each period, she picks the informativeness of a signal and sells it by offering a spot contract to the agent, and we assume that the principal has limited commitment power in the sense that she is unable to commit to contracts beyond the current period.    
 
 The principal can only make money as long as the search continues. However, the information provided and the price charged can potentially affect the duration of search. The dynamic nature of the agent’s search gives rise to two important considerations for the principal's profit maximization. First, for a fixed sequence of future contracts, she must contend with an \emph{intra-temporal} trade-off between surplus extraction and persuasion: the principal can earn a higher profit today by selling a Blackwell more informative signal, but a more informative signal can lower the probability with which the agent is persuaded  to continue his search. Second, contracts offered in subsequent periods have \emph{inter-temporal} impacts on the trade-off between extraction and persuasion in the current period: higher prices or less informative signals in the future lower the agent's value from continued search, thereby reducing his willingness to continue searching today (and thus the principal's ability to persuade him). We study stationary equilibria in which the same spot contract resolves the intra-temporal trade-off in each period while also taking into account its inter-temporal impact.

Our main result shows that the principal, despite lacking long-term commitment power and being constrained to offer stationary contracts, does as well as a principal who could commit to any dynamic selling mechanism. In particular, we show that she extracts the socially efficient surplus in any stationary equilibrium. Efficiency implies that the agent searches as if he receives full information for free; he terminates his search if and only if a good's quality exceeds the McCall reservation value---the same reservation value used by the agent in \cite{mcc70} in which he observes a fully informative signal for free in each period. In contrast, surplus extraction implies that the agent is no better off than in autarky, i.e., without any information. Therefore, he would prefer to terminate his search for any good whose expected quality exceeds the \emph{autarky reservation value}---the reservation value the agent would use without any signal of a good's quality. The McCall reservation value exceeds the autarky reservation value, implying that the agent searches longer than he would otherwise prefer in equilibrium.

Surplus extraction is intuitive in our setting; the agent lacks any bargaining power because the principal designs the contract in each period. Furthermore, in a stationary equilibrium, the agent cannot use the threat of low continuation payoffs to prevent profitable deviations by the principal. 

What is more surprising is that a principal with limited commitment can nevertheless generate the socially efficient surplus and simultaneously extract it.\footnote{Our result is reminiscent to that of \cite{bergemannbrooksmorris} who, in a monopoly screening problem, show that it is possible to design information that generates socially efficient trade while simultaneously allowing consumers to extract all surplus in excess of the seller's no-information profit.} In order to achieve this outcome in a stationary equilibrium, the principal offers a spot contract with two prominent and interdependent features. First, the principal  persuades the agent against an early exit from the search market by pooling goods with qualities below the McCall reservation value into a ``fail" signal realization.\footnote{\label{footnote_cost}Such pooling signals are in line with the pre-employment testing studied in \cite{aut08}. They document that while the assessments generate a rich set of test scores, the information broker pools low performers into an implicit ``fail" category (the probability of being hired conditional on inclusion in this category is 0.08\%). This observation also suggests that it is not costly for the information broker to generate more informative signals;  otherwise, the practice of pooling realizations ex-post to generate a coarser signal would be costlier than simply producing the coarser signal to begin with. Our model reflects this observation by assuming that the principal's signal production is costless. In \hyperref[stationary]{\Cref{stationary}}, we discuss  implications of costly signal production.} Pooling is necessary for full surplus extraction; if the principal offers the agent a signal that reveals any additional information about a good whose quality falls between the autarky and McCall reservation values, the agent would consume the good, thereby terminating his search inefficiently early. The implementability of this pooling signal crucially hinges on the agent's incentives to continue his search when he sees ``fail" and to stop his search otherwise. It is intuitive to see that the agent is persuaded to stop his search when he does not observe a ``fail." We show that by pooling goods whose qualities fall between the autarky and McCall reservation values (sufficiently high quality goods for which the agent would be willing to stop his search) with goods whose qualities fall below the autarky reservation value (low quality goods for which he would rather continue searching), the agent is persuaded to continue searching when he observes a ``fail."

Second, the per-period price the principal charges for this pooling signal is equal to the agent's willingness-to-pay (WTP) for it, implying that the agent's participation constraint binds. In general, this per-period price is a fraction of the expected total profit the principal derives. Consequently, the principal extracts the entire surplus through dribs and drabs over a long period of time.

Our model yields stark findings of full surplus extraction, and the principal neither gains from considering more complex history-dependent contracts nor has any value for long-term commitment. In a number of extensions, we test the robustness of these findings to differences present in various search markets. Notably, our results on surplus extraction and efficiency are qualitatively unchanged in an alternative model in which the agent pays the principal only upon successfully completing his search, as is typical in real-estate. In a market with competing principals, the essentially unique stationary equilibrium outcome resembles that of classical Bertrand competition, and obtains efficiency of search but with the agent capturing the surplus \citep{MPH}.

Our findings are more nuanced in markets where some information is available even without contracting with the principal. For example, a job candidate with a  college degree is informationally differentiated from a job candidate without a college degree, or a job candidate with a history of employment in a particular field may have a different reputation than a newcomer. Formally, an exogenous signal of each good's quality is publicly observed immediately upon sampling and prior to contracting with the principal. The principal offers a spot contract that can depend on the realization of the public signal (e.g., providing a different skill-assessment to those with versus without college degrees). This model nests both our base model (if the public signal is completely uninformative) and the model of \cite{mcc70} (if the public signal is completely informative). 
  
We show that the presence of public information in a persuaded search game may result in an ex-post lower-quality match for the agent. From a policy perspective, our finding implies that blinding the players to the public signal (e.g., hiding candidates' resumes) can improve efficiency in the search market. Empirical evidence supports this claim; \cite{hof18} find that searching firms who rely on pre-employment testing along with resume screening hire less productive workers than those who rely only on pre-employment testing.  We provide tight conditions for when inefficiency can occur: either when interim beliefs cause the agent to reject contracting with the principal because he is too \emph{optimistic}, or when interim beliefs have a \emph{thin left tail} that require the principal to offer a distorted signal in order to persuade the agent to continue his search.

\subsection*{Related Literature}
Our paper marries the literature on sequential search with the more recent literature on information design. There is a small but growing body of work at the intersection of these two literatures. In the context of random search markets, \cite{do22} and \cite{hu22} consider a centralized information design problem in which a planner shapes the price competition between firms by controlling the information about the quality of goods that the consumers observe, while \cite{boa19} and \cite{whi21} consider a decentralized information design problem wherein each firm competes by choosing how much information it discloses to consumers. \cite{au2023attraction} similarly study firm competition in prices and information but in the context of directed search markets.\footnote{\cite{an06} and \cite{ch19} study information design for a search good, whose true value is revealed upon sampling. In these models, the information designer seeks to persuade an agent to engage in search. Our model crucially differs in that we study search for an experience good whose true value is revealed only upon consumption following the termination of search. Therefore, our principal seeks to persuade the agent to continue his search.}

Our work differs from these papers in two ways: First, in all the aforementioned  papers, information itself is not for sale; it is merely a tool for impacting the price and the selling probability of a separate good. In contrast, the main focus of our paper is precisely the design and pricing of information as a good. Second, 
 the relationships in all these papers are short term: if a consumer and a firm fail to trade, both move on to new trading partners. In contrast, a key feature of our framework is that the relationship between the agent and the principal is long term, which allows us to study the inter-temporal effects of repeated contracting. These dynamic effects also differentiate our paper from the related literature on the design and pricing of information goods \citep{adm86, adm90, es07,ber15, ber18}  and the literature on certification \citep{mat85,sha94,liz99, ali22}, which all focus on static problems. 
 
 Information design with long-term relationships has also been studied by \cite{hor16}, \cite{orl20}, and \cite{ely20}. A common feature of these papers is a persistent state that is unobserved by the agent, which leads to a non-stationary contractual environment because the agent becomes more informed over time. In all three, commitment power plays an important role in equilibrium outcomes. Our model does not feature cumulative learning by either the agent or the principal as the underlying environment is stationary, and we show that the principal cannot improve on her stationary equilibrium payoffs with commitment power.

The rest of the paper is organized as follows: We describe our model in \hyperref[model]{\Cref{model}} and introduce analytical tools in \hyperref[prelim]{\Cref{prelim}}. We then present our main result in \hyperref[stationary]{\Cref{stationary}} and extend our analysis to include public signals in \hyperref[public]{\Cref{public}} before concluding in \hyperref[conclude]{\Cref{conclude}}. Omitted proofs are located in the \hyperref[appendix]{Appendix}, and we include extensions of our model to allow for flow search costs, unequal discount factors, and a pay-at-the-end contracting scheme in the Online Appendix.

\section{Model}\label{model}
We consider an information design problem that is embedded in a sequential search problem without recall. The game involves an agent (he) who searches within the market for a good to consume, and a principal (she) who controls the information the agent observes while he searches.

The agent has unit demand and his ex-post payoff from consuming a good equals the good's quality. In each period $t=1, 2,\ldots$ that the agent remains in the market, he samples a good of quality $\theta_t$ which is identically and independently distributed according to an absolutely continuous distribution $F$ on an interval $[\underline\theta, \bar\theta]\triangleq \Theta$ with $\infty>\bar\theta>\underline \theta>-\infty$. We denote the prior expected quality by  $\mathbb{E}_F[\theta]\triangleq m_\varnothing$ and assume that $m_\varnothing>0$.\footnote{Our results do not depend on the sign of $m_\varnothing$. We make this assumption to simplify the exposition.}

Neither the agent nor the principal observe the sampled good's true quality. Instead, the agent can acquire information by purchasing an informative signal that is designed and sold by the principal. Formally, the principal offers a contract $\langle p_t, (\pi_t, S_t)\rangle$ consisting of a price $p_t\in\mathbb{R}_+$ and a signal $(\pi_t, S_t)$ where $\pi_t:\Theta\to \Delta(S_t)$ is a mapping from a good's true quality to a probability distribution over a compact set of possible signal realizations $S_t$. The principal can design any signal in any given period but she has limited commitment, i.e.,  she cannot commit to contracts beyond the current period. 

Given a contract $\langle p_t, (\pi_t, S_t)\rangle$, the agent either accepts or rejects it. If he accepts the contract $(a_t=1)$, he pays $p_t$ and both players observe a common signal realization $s_t\in S_t$ distributed according to $\pi_t(\theta_t)$. If he instead rejects the contract $(a_t=0)$, he pays nothing and no signal is observed. Finally, the agent decides whether to consume the sampled good and stop searching $(d_t=1)$, in which case the game ends, or to continue searching $(d_t=0)$, in which case the game continues on to $t+1$.\footnote{Our model does not allow the agent to  recall a previously sampled good. Consider a richer model in which the agent can recall goods and the principal observes whenever a good is recalled. For any stationary equilibrium of the base game we study, there exists an equilibrium of this richer model in which 1) the principal provides no new information for any recalled good, 2) the agent never recalls on path, and 3) both players receive the same payoffs. } \autoref{fig:timing} depicts the timing of events within a period and \autoref{tab:payoffs} presents the per-period payoffs. The agent and the principal have a common discount factor $\delta\in (0,1)$.

\begin{figure}[ht]
\centering 
\begin{tikzpicture}[scale=0.5]
\tikzset{
    position label/.style={
             text depth = 1ex
    }
}
    \foreach \x in {-5, -1}
      \draw (\x cm,3pt) -- (\x cm,-3pt);
\draw[-] (-6,0) node[left]{\footnotesize{ $t$}}--(2,0);
\draw[-] (2,0)--(3.5,3);
\draw[-] (2,0)--(3.5,-3);
\draw[-] (3.5,-3)--(7.5,-3);
\draw[-] (3.5,3)--(7.5,3);
\draw[-] (7.5,-3)--(12,0);
\draw[-] (7.5,-3)--(12,-5);
\draw[-] (7.5,3)--(12,5);
\draw[-] (7.5,3)--(12,0);
\draw[->] (12,0)--(13,0)node[right]{\footnotesize{ ${t+1}$}};

\node [position label, align=center, below=3pt] (fend) at (-5,0) {\footnotesize{Good} \\ \footnotesize{ of quality}\\ \footnotesize{$\theta_t\sim F$}\\ \footnotesize{sampled}};
\node [position label, align=center, below=3pt] (fend) at (-1,0) {\footnotesize{Principal} \\ \footnotesize{offers}\\ \footnotesize{ $\langle p_t, (\pi_t, S_t)\rangle$}};
\node [position label, rotate=63, align=center, above] (fend) at (3,1.5) {\footnotesize{$a_t=1$}};
\node [position label, rotate=297, align=center, below] (fend) at (2.8,-1.5) {\footnotesize{$a_t=0$}};

\draw[](5.5,3.1)--(5.5,2.9);
\draw[](5.5,-3.1)--(5.5,-2.9);

\node [position label, align=center, above] (fend) at (5.5,3) {\footnotesize{Observe}\\ \footnotesize{$s_t\sim\pi_t(\theta_t)$}};
\node [position label, align=center, below] (fend) at (5.5,-3) {\footnotesize{Observe}\\ \footnotesize{nothing}};

\node [position label, align=center, right] (fend) at (12,5) {\footnotesize{End}};
\node [position label, align=center, right] (fend) at (12,-5) {\footnotesize{End}};

\node [position label, rotate=335, align=center, below] (fend) at (10,-4) {\footnotesize{$d_t=1$}};
\node [position label, rotate=35, align=center, above] (fend) at (10,-1.8) {\footnotesize{$d_t=0$}};

\node [position label, rotate=25, align=center, above] (fend) at (10,3.8) {\footnotesize{$d_t=1$}};
\node [position label, rotate=325, align=center, below] (fend) at (10.2,1.35) {\footnotesize{$d_t=0$}};
  \end{tikzpicture}
  \caption{Timing of events}
  \label{fig:timing}
\end{figure}
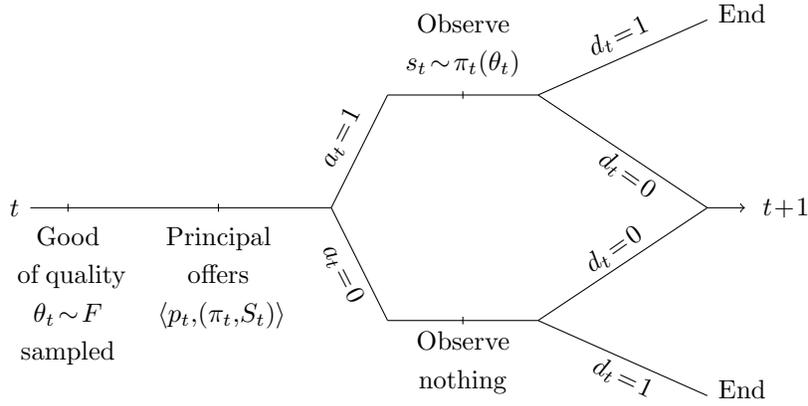 

\begin{table}[ht]
\begin{center}
    \begin{tabular}{c|c|c}
         & $a_t=1$& $a_t=0$ \\[8pt] \hline \hline
        $d_t=1$ & $\theta_t-p_t, p_t$& $\theta_t, 0$\\[8pt] \hline
        $d_t=0$& $-p_t, p_t$ & $0,0$
    \end{tabular}
    \caption{Period-$t$ payoffs for the agent and the principal, respectively.}
    \label{tab:payoffs}
    \end{center}
\end{table}

We investigate stationary (perfect Bayesian) equilibrium outcomes of this game, which are  characterized by a strategy specifying contract proposal for the principal, a strategy specifying contract acceptance and search stopping for the agent, and a belief-updating process such that $(i)$  strategies are stationary, i.e.,  history-independent both on and off the equilibrium path, $(ii)$ strategies are sequentially rational, and $(iii)$ beliefs are derived by Bayes rule. We provide a more formal definition after we define the contract space. 

We make the following assumption moving forward, which we later show is necessary and sufficient for non-trivial contracting between the principal and the agent in equilibrium:
\begin{assumption}\label{ass:1}
$\underline \theta<\delta m_\varnothing$.
\end{assumption}

\hyperref[ass:1]{\Cref{ass:1}} is satisfied for all $\delta \in(0,1)$ if $\underline \theta\leq  0$. However, if $\underline \theta> 0$, the assumption places further restrictions on the values of the discount factor.

\section{Preliminary Analysis}\label{prelim}
\subsection{Signals as posterior-mean distributions}
Suppose the principal proposes a contract $\langle p, (\pi, S)\rangle$ in a given period and the agent rejects the contract. In this case, the agent's expected payoff within the period is zero if he continues his search and $m_\varnothing$ if he ends his search. In contrast, if the agent accepts the contract, he then observes a signal realization $s\in S$ and  updates his belief about the good's quality from his prior $F$ to a posterior  $F_s$ using Bayes rule. The posterior mean would then be $\mathbb{E}_{F_s}[\theta]$. In this case, the agent's expected payoff within the period is $-p$ if he continues his search and  $\mathbb{E}_{F_s}[\theta]-p$ if he ends his search.

Because the signal $(\pi, S)$ impacts the agent's payoffs only through a good's expected quality, it suffices for our purposes to consider the distribution over posterior means that is induced by the signal. For example, a fully informative signal induces the distribution $F$ while an uninformative signal induces the degenerate distribution $G^\varnothing$ given by
\[
G^\varnothing(m)=\left\{\begin{array}{ccc}
     0 &\mbox{if} & m<m_\varnothing  \\
     1 &\mbox{if} & m\geq m_\varnothing 
     \end{array}\right..
\]

From \cite{bla53} and \cite{rot70}, each signal $(\pi, S)$ induces a distribution over posterior means $G^\pi:\Theta\to[0,1]$ that is a mean-preserving contraction of $F$, i.e., 
\[
\int^{\bar \theta}_x G^\pi(m)dm\geq \int^{\bar \theta}_x F(m)dm
\]
for all $x\in\Theta$ with equality at $x=\underline\theta$. Conversely, each distribution $G$ that is a  mean-preserving contraction of $F$ can be induced by some signal. Thus, without loss of generality, we assume that the principal can propose any contract of the form $\langle p, G \rangle \in \mathbb{R}_+\times \mathcal{G}(F)$, where $\mathcal{G}(F)$ is the set of mean-preserving contractions of $F$. Notice that $G^\varnothing$ is a mean-preserving contraction of any $G\in \mathcal{G}(F)$.
\subsection{Signals as convex functions}

Given some posterior-mean distribution $G\in\mathcal{G}(F)$, let $c_G:\Theta\to \mathbb{R}_+$ be defined as
\begin{align*}
 c_G(x)&\triangleq\int^{\bar\theta}_{x}(1-G(m))dm.
\end{align*}

To understand what $c_G(x)$ represents, suppose the agent has a guaranteed payoff of $x\in\Theta$ from some ``outside option." The agent's expected surplus over the outside option from a random draw of distribution $G$ is given by
\begin{align*}
\left(\mathbb{E}_G[m|m\geq x]-x\right)\left(1-G(x)\right)&=\int^{\bar\theta}_x(m-x)dG(m)\\[6pt]
&=c_G(x),
\end{align*}
where the final equality follows from integration by parts.\footnote{\cite{gen16} study a related function, $\tilde c_G(x)\triangleq \int_{\underline\theta}^xG(m)dm$ and use it to provide an alternative characterization for the set of posterior-mean distributions. Since $G$ is a mean-preserving contraction of $F$, the graph for $c_G$ is a reflection of $\tilde c_G$ along the vertical line passing through $\mathbb{E}_F[\theta]$. We proceed with the alternative function $c_G(\cdot)$--which appears in many search  models \citep{mcc70, wei79, wol86, dog18}--because it is more convenient in characterizing the value of search for a good whose quality exceeds a particular value $x$.} 

\begin{remark}
\label{remark:1}
For any $G\in\mathcal{G}(F)$, 
\begin{enumerate}[$(a)$]
\item $c_G(\underline\theta)=m_\varnothing-\underline \theta$,
    \item $c_G(\bar\theta)=0$, 
    \item $c_G$ is continuous, weakly decreasing, and  convex, and
    \item $c_G$ is right differentiable with  $\partial_+ c_G(x)=G(x)-1$.
\end{enumerate}
\end{remark}

Intuitively, if the payoff from the outside option is $\underline \theta$, then a random draw from $G$ yields a higher payoff than the outside option with probability one, so $c_G(\underline\theta)=m_\varnothing-\underline\theta$. In contrast, if the payoff from the outside option is $\bar \theta$, then there is zero probability that a random draw from $G$ yields a higher payoff than the outside option, so $c_G(\bar\theta)=0$. Additionally, consider two outside options with payoffs $x$ and $x'$ with $x'>x$; a draw $m$ from $G$ is more likely to exceed $x$ than $x'$, and even if $m$ exceeds both, $m-x>m-x'$. These two properties combined imply that $c_G(x)$ is a weakly decreasing and convex function of $x$. Convexity further implies that $c_G(x)$ is left- and right-differentiable for all $x\in\Theta$. It suffices to consider only the right derivative for our purposes.

Notice that if $G'$ is a mean-preserving contraction of $G''$, then $c_{G'}\leq c_{G''}$ pointwise. Hence, for all $G\in\mathcal{G}(F)$,  $c_{G^\varnothing}\leq c_G\leq c_F$ pointwise. 

\subsection{Agent's search problem}\label{autarky_section}
Before we analyze the principal-agent game, we first consider a hypothetical setting in which there is no principal, and the agent  observes a realization from some posterior-mean distribution $G\in \mathcal{G}(F)$ for free in each period. This setting is instructive in characterizing when there is scope in our model for contracting between the principal and the agent.  

Let $u(G)$ be the agent's expected payoff in this search problem, which is also his continuation payoff due to the stationary setting. The agent  searches until he finds a good  whose expected quality exceeds his discounted continuation value $\delta u(G)$. Thus, the agent's payoff $u(G)$ is the unique solution to the fixed point (as a function of $u$)
\begin{align*}
 \label{eq:1}
\tag{1}
 u&=\int_\Theta \max\{m, \delta u\}dG(m).
\end{align*}

As the integrand is a convex function of $m$, the agent benefits from a ``more dispersed" distribution of posterior means. Formally, $u(G')\leq u(G'')$ whenever $G'$ is a mean-preserving contraction of $G''$. Thus, the highest payoff that can be generated in the search market is $\bar u\triangleq u(F)>0$ while the lowest payoff that can be generated is $\underline u\triangleq u(G^\varnothing)=m_\varnothing$, which is the payoff from accepting any sampled good in the current period. We refer to $\bar u$ as the McCall payoff---the payoff generated by the agent in \cite{mcc70}---and to $\underline u$ as the autarky payoff. For any $G\in\mathcal{G}(F),$ $u(G)-\underline u$ represents the surplus generated from searching given $G$, and we refer to $\bar u-\underline u$ as the full surplus.

If $\underline u=\bar u$, we say the agent \textit{never searches} because  it is optimal in this case for the agent to stop searching and consume the first good he samples after observing any realization from any $G\in\mathcal{G}(F)$. Intuitively, an agent never searches when the lowest quality good yields a sufficiently high payoff; even if the agent knew he would sample a good of quality $\bar \theta$ in the next period, he would prefer to accept a good of quality $\underline \theta$ today. In other words, if $\underline u=\bar u$ then the agent would have no value for information, making the contracting game between the agent and the principal trivial. \autoref{lemma:1} presents a sufficient and necessary condition for when the agent never searches.

\begin{lemma}
\label{lemma:1}
The agent never searches if and only if $\underline \theta\geq \delta m_\varnothing$. 
\end{lemma}

In light of \autoref{lemma:1}, the restriction in \hyperref[ass:1]{\Cref{ass:1}} that  $\underline \theta< \delta m_\varnothing$  implies that the agent's continuation value is a non-constant function of distribution $G$.  To quantify this dependence, we rewrite the fixed-point problem in \eqref{eq:1} as

\begin{align*}
    u&=\delta u+\int^{\bar\theta}_{\delta u}(m-\delta u)dG(m)\\[6pt]
    \label{eq:1'}
    \tag{$1'$}
    &=\underbrace{\delta u}_{\substack{\text{reservation}\\\text{value}}} + \underbrace{c_G(\delta u)}_{\substack{\text{added value}\\\text{from search}}}.
\end{align*}

Given $G\in\mathcal{G}(F)$, let $r(G)\triangleq \delta u(G)$ denote the agent's reservation value---the expected quality which makes the agent indifferent between stopping his search and continuing when he observes $G$ for free in subsequent periods. We can re-purpose \eqref{eq:1'} so as to characterize $r(G)$ as the unique solution to the fixed point  (as a function of $r$)
\begin{align*}
 r&=\delta\big( r+c_G(r)\big)\\[6pt]
 \label{eq:2}
 \tag{2}
 &=\left(\frac{\delta}{1-\delta}\right)\cdot c_G(r).
\end{align*}
 The fixed point in \eqref{eq:2} is particularly amenable to a geometric representation as shown in \autoref{fig:fix}. Since $c_{G'}\leq c_{G''}$ pointwise whenever $G'$ is a mean-preserving contraction of $G''$, the highest reservation value is $\bar r\triangleq r(F)$ and the lowest reservation value is $\underline r\triangleq r(G^\varnothing)$. We refer to $\bar r$ as the McCall reservation value and to $\underline r$ as the autarky reservation value.\footnote{If the agent has no information, his posterior mean is  $m_\varnothing$ with probability one, which exceeds the autarky reservation value $\underline r$ for all $\delta\in (0,1)$ as in \autoref{fig:fix}. Thus, with no information, the agent immediately stops his search and consumes the first sampled good.}

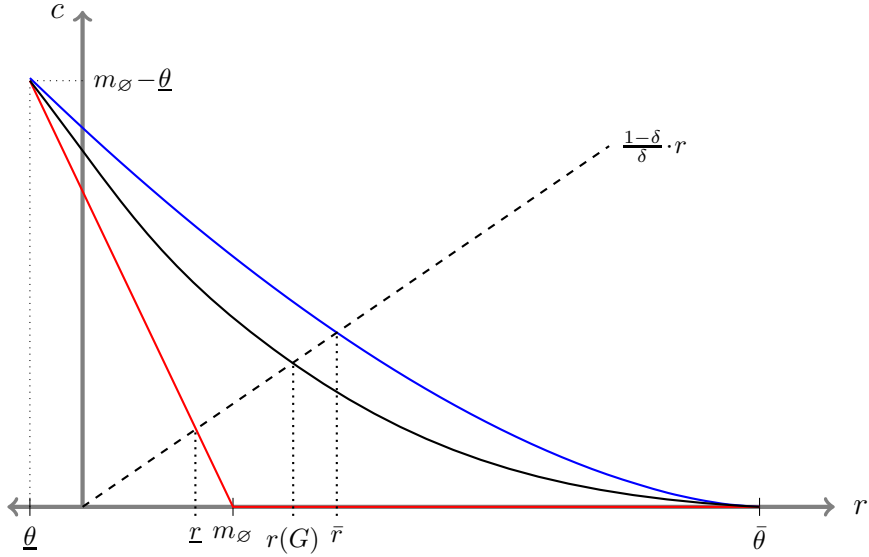
\begin{figure}[ht]
\begin{center}
\begin{tikzpicture}[xscale=10, yscale=6]
\draw [<->, help lines, ultra thick] (0,0) -- (1.1,0);
\draw [->, help lines, ultra thick] (0.1,0) --(.1,1.1);\node[right=3pt] at (1.1,0) {$r$};
\node[left=3pt] at (0.1, 1.1) {$c$};

\draw [] (1,-.02)node[below]{\footnotesize{$\bar\theta$}}--(1,.02);
\draw [] (.03,-.02) node[below]{\footnotesize{$\underline\theta$}}--(.03,.02);
\draw [dotted] (0.1, .945)node[right]{\footnotesize{$m_\varnothing -\underline\theta$}}--(0.03,.945);
\draw [dotted] (.03, 0)--(.03,.945);
\draw [] (.3,-.02) node[below]{\footnotesize{$m_\varnothing$}}--(.3,.02);

\draw[blue, domain=0.03:1,smooth,variable=\x, thick] plot ({\x},{(1-\x)^1.65});
\draw[thick, red](.03, .945)--(.3,0)--(1,0);
\draw[ thick] plot [smooth, tension=2] (0.03,.945)to [in=130, out=295](.37, .33) to [in=175, out=310](1,0);

\draw[thick, dashed] (0.1,0)--(.8,.8)node[right]{\footnotesize{$\frac{1-\delta}{\delta}\cdot r$}};

\draw [dotted, thick] (.38,-.02) node[below]{\footnotesize{$r(G)$}}--(.38,.32);
\draw [dotted, thick] (.25,-.02) node[below]{\footnotesize{$\underline r$}}--(.25,.172);
\draw [dotted, thick] (.438,-.02) node[below]{\footnotesize{$\bar r$}}--(.438,.39);
\end{tikzpicture}
\captionsetup{oneside,margin={0cm,0cm},justification=justified, singlelinecheck=false}
\caption{The dashed line has a slope of $(1-\delta)/\delta$. The red curve is $c_{G^\varnothing}$, the black curve is $c_G$ for some arbitrary $G\in\mathcal{G}(F)$, and the blue curve is $c_F$. The intersection of each respective curve with the dashed line represents the solution to the fixed point problem in \eqref{eq:2}.}
\label{fig:fix}
\end{center}
\end{figure}

\subsection{Benchmark: Long-term Contracts and Full Surplus Extraction}\label{sec:benchmark}
In our model, there is no private information on the part of the agent or the principal. Therefore, a principal with long-term commitment power can extract the full surplus of search in equilibrium. For example, the principal could propose a sequence of contracts $\langle p_t, G_t\rangle_{t\geq 1}$ such that $G_t=F$ for all $t\geq 1$, $p_t=0$ for all $t>1$, and $p_1=\bar u-\underline u$. If the agent rejects this sequence of contracts (at $t=1$) then the principal commits to ending the relationship, i.e., by offering a sequence of contracts $\langle p_t, G_t\rangle_{t\geq 1}$ such that $G_t=G^\varnothing$ and $p_t=0$ for all $t>1$. Such a sequence of contracts extracts the full surplus from the agent and is equivalent to ``selling the information production technology" to the agent. However, without long-term commitment power,  the principal would deviate away from charging $p_t=0$  in periods $t>1$, and thus, the agent would never accept the contract. 

Despite the lack of private information in our setting, it is not a priori evident that the principal can achieve as a high a payoff without long-term commitment power; after all, the set of perfect Bayesian equilibrium outcomes could be a strict subset of the set of Nash equilibrium outcomes.

\section{Stationary Equilibrium}\label{stationary}
In this section, we characterize stationary equilibrium outcomes under limited commitment. To that end, we make two observations that simplify our task. First,  it suffices for our purposes to consider only pure strategies because, since all players are risk neutral and the space of contracts is convex, the principal does not gain from randomizing over contracts.\footnote{Additionally, to sustain an equilibrium, we also assume that whenever the agent is indifferent between any two actions on the equilibrium path, he chooses the one that maximizes the principal's payoff. } Second,  we can restrict attention to contracts that induce the agent to accept since the agent rejecting a contract is payoff-equivalent for all players to the agent accepting $\langle 0, G^\varnothing\rangle$.

\begin{definition}
\label{def:1}
\normalfont A stationary equilibrium is defined by a pair of agent-principal continuation values $(U, V)\in \mathbb{R}_+^2$ and a contract $\langle p,G\rangle\in\mathbb{R}_+\times \mathcal{G}(F)$ such that in each period and after every history:
\begin{enumerate}
    \item The agent stops searching and consumes a good of expected quality $m\in\Theta$ if and only if 
    \[
     \label{eq:os}
    \tag{OS}
    m>\delta U,
    \]
    \item The agent accepts a contract $\langle \widehat p,\widehat G\rangle\in\mathbb{R}_+\times\mathcal{G}(F)$  if and only if
    \begin{align*}
        \underbrace{\delta U+c_{\widehat G}(\delta U)-\widehat p}_{\text{payoff from accept} }& \geq  \underbrace{\delta U+c_{ G^\varnothing}(\delta U)}_{\text{payoff from reject}} \\[8pt]
        \label{eq:pc}
    \tag{PC}
        \Leftrightarrow   c_{\widehat G}(\delta U)-c_{G^\varnothing}(\delta U)& \geq \widehat p,
    \end{align*}
       
\item The principal proposes contract $\langle p,G\rangle$ to maximize her profit, i.e.,
\[
\label{eq:pm}
\tag{PM}
 \langle p,G\rangle\in \argmax_{\langle \widehat p, \widehat G\rangle\in\mathbb{R}_+\times \mathcal{G}(F)}\hspace*{.2em} \widehat p+ \widehat G(\delta U) \delta V \hspace*{.5em} \text{ subject to \eqref{eq:pc}},
\]   
\item Payoffs are self-generating, i.e., 
\[
\label{eq:sga}
\tag{SG-A}
U=\delta U+c_{G}(\delta U)-p 
\] 
    and
    \[
    \label{eq:sgp}
\tag{SG-P}
V= p + G(\delta U) \delta V.
\] 
\end{enumerate}
\end{definition}

To gain some intuition for the economic forces at play, let us first consider the principal's intra-temporal trade-off\textemdash the trade-off the principal faces when proposing a contract in any given period $t\geq 1$ while holding the continuation values $(U,V)$ fixed. By selling a signal that induces a posterior-mean distribution  $\widehat G$ in period $t$, the principal can at most earn a revenue of $c_{\widehat G}(\delta U)-c_{G^\varnothing}(\delta U)$, which is captured by \eqref{eq:pc}. Additionally, she earns a payoff of $\delta V$ by persuading the agent to continue searching in period $t$ with probability $\widehat G(\delta U)$, which is derived from \eqref{eq:os}.  The principal can extract more surplus in period $t$ by selling a more informative signal, but a more informative signal can lead the agent to stop searching with a higher probability.  Thus, the principal trades off surplus extraction and persuasion as captured by the constrained profit maximization problem \eqref{eq:pm}.  

However, we cannot consider the intra-temporal trade-off of any given period in isolation. This is evident from the fact that the surplus the principal can extract and the probability with which she can persuade the agent to continue searching in any given period both depend on the continuation values, which are derived from the contracts the principal offers in subsequent periods. In a stationary equilibrium, the same contract must resolve the intra-temporal trade-off in each period while also taking into account how these trade-offs are inter-temporally linked. Formally, a contract $\langle p, G\rangle$ supports a stationary equilibrium if it solves \eqref{eq:pm} for given a pair of continuation values $(U,V)$, which are themselves generated from $\langle p, G\rangle$ according to \eqref{eq:sga} and \eqref{eq:sgp}.

Our main result below establishes not only the existence of a stationary equilibrium but also that all stationary equilibria are payoff equivalent. In any stationary equilibrium, the McCall payoff $\bar u$ is generated but the agent's payoff is the same as his autarky payoff $\underline u$. The full surplus $\bar u-\underline u$ is extracted by the principal through dribs and drabs by charging a per-period price that is a (possibly small) fraction of the overall value of search.\footnote{As we show in point $(iii)$ of \autoref{prop:1}, the per-period price is strictly below $\bar u-\underline u$ for any $\delta>0$. Additionally, as $\delta\to 1$, the per-period price converges to zero while the duration of search goes to infinity. } Full surplus extraction implies that the agent searches as if he observes a fully informative signal  for free, even though his payoff is as if he observes no information.

In order to state our main result, we introduce a particular family of ``pass/fail" signals: for  $x\in\Theta$,  let $(\pi(x), \{\text{pass, fail}\})$ be given by 
 \[
\pi(\text{pass}|\theta;x)=\left\{\begin{array}{ccl}
     0 &\mbox{if} & ~~~\theta\leq  x \\
     1 &\mbox{if} & ~~~\theta> x
     \end{array}\right..
\] 
We will refer to such signals simply as $\pi(x)$ for  $x\in\Theta$, and denote the posterior-mean distribution it induces by $G^{\pi(x)}\in\mathcal{G}(F)$. Notice that both $\pi(\underline \theta)$ and $\pi(\bar \theta)$  induce $G^{\pi(\underline\theta)}=G^{\pi(\bar\theta)}=G^\varnothing$ as almost all goods pass for the former and  all goods fail for the latter. On the other hand, when $x$ is an interior point, i.e., $x\in \text{int}(\Theta)$, $\pi(x)$ induces a binary posterior-mean distribution given by 
 \[
 \label{eq:3}
\tag{3}
G^{\pi(x)}(m)=\left\{\begin{array}{ccl}
     0 &\mbox{if} & ~~~m<\mathbb{E}_F[\theta|\theta\leq  x]  \\
       F(x) &\mbox{if} &~~~ \mathbb{E}_F[\theta|\theta\leq  x]\leq m<\mathbb{E}_F[\theta|\theta>  x]  \\
     1 &\mbox{if} & ~~~m\geq \mathbb{E}_F[\theta|\theta>  x]
     \end{array}\right..
\]
We call $G^{\pi(x)}$ in \eqref{eq:3} a binary distribution as it has mass at only two values: $\mathbb{E}_F[\theta|\theta\leq  x]$ and $\mathbb{E}_F[\theta|\theta>  x]$. This family of distributions plays an important role in characterizing stationary equilibria, and we discuss its properties imminently as part of the argument for \autoref{prop:1}.

\begin{theorem}\label{prop:1}\
A stationary equilibrium exists. A pair of agent-principal continuation values $(U, V)\in \mathbb{R}_+^2$ and a contract $\langle p,G\rangle\in\mathbb{R}_+\times \mathcal{G}(F)$ constitute a stationary equilibrium if and only if
\begin{enumerate}[$(i)$]
    \item  $U=\underline u$, 
    \item  $V=\bar u-\underline u$, 
        \item  $p=(\bar u-\underline u)\big(1-\delta F( \bar r)\big)$,
    \item $G(\underline r)=F( \bar r)$, and
    \item $c_G\geq c_{G^{\pi(\bar r)}}$ pointwise with $c_G(\underline r)= c_{G^{\pi(\bar r)}}(\underline r)$.
\end{enumerate}
\end{theorem}

We present a formal proof of this result in the appendix. Here we sketch the argument that the principal extracts the full surplus in a stationary equilibrium, i.e., points $(i)$ and $(ii)$ of the theorem must be satisfied in any stationary equilibrium.

First, \eqref{eq:pc} must bind in each period following any history: $c_{G}(\delta U)-c_{G^\varnothing}(\delta U)= p$. Intuitively, the agent has no bargaining power because he receives a take-it-or-leave-it offer in every period, and because his continuation strategy must be history-independent in a stationary equilibrium. Thus, if \eqref{eq:pc} did not bind in a given period, the principal could profitably deviate by increasing the price without affecting the agent's continuation value (and hence stopping decision). Given the participation constraint binds in each period, the agent must receive his autarky payoff, i.e. $U=\underline u.$

Having pinned down the agent's continuation value, we take a geometric approach to analyze the principal's intra-temporal trade-off between surplus extraction and persuasion in \eqref{eq:pm}. Because the agent has a binary decision regarding his search after observing the signal realization--either to continue or to terminate--it suffices to consider binary posterior-mean distributions to solve \eqref{eq:pm}. To that end, consider an arbitrary binary distribution $G'$, where $c_{G'}$ is depicted by the black piece-wise linear convex curve in \autoref{fig:44a}. The vertical difference between $c_{G'}$  and $c_{G^\varnothing}$ at $\delta U=\underline r$ captures the extracted surplus in each period\textemdash it is the price $p'$  which makes the agent's participation constraint bind\textemdash while the slope of $c_{G'}$ at $\underline r$ captures persuasion\textemdash by point $(d)$ of \autoref{remark:1},  $\partial_+c_{G'}(\underline r)+1=G'(\underline r)$ is the probability with which the agent continues searching.

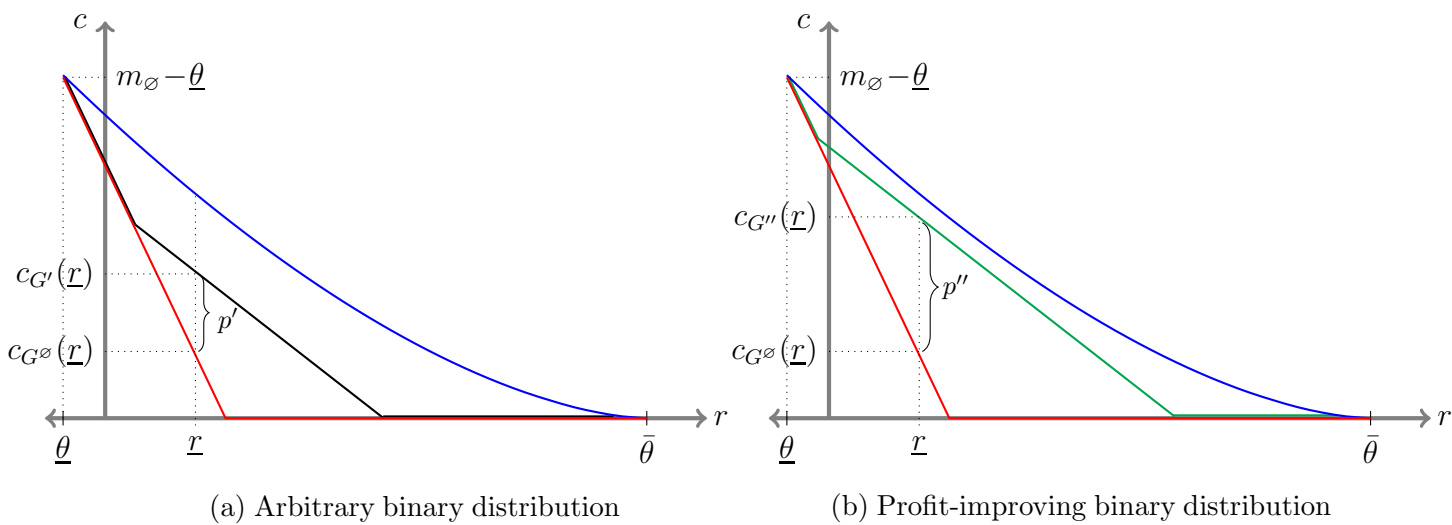
\begin{figure}[ht]
\begin{center}
\hspace{-50mm}
\captionsetup[subfigure]{oneside,margin={1cm,-4cm}}\begin{subfigure}[t]{0.4\textwidth}
\centering
\begin{tikzpicture}[xscale=8, yscale=4.8]
\draw [<->, help lines, ultra thick] (0,0) -- (1.1,0);
\draw [->, help lines, ultra thick] (0.1,0) --(.1,1.1);\node[right=3pt] at (1.08,0) {$r$};
\node[left=3pt] at (0.1, 1.1) {$c$};

\draw [] (1,-.02)node[below]{$\bar\theta$}--(1,.02);
\draw [] (.03,-.02) node[below]{$\underline\theta$}--(.03,.02);
\draw [dotted] (0.1, .945)node[right]{$m_\varnothing -\underline\theta$}--(0.03,.945);
\draw [dotted] (.03, 0)--(.03,.945);

\draw[ thick](.0327, .947)--(.15,.5365);
\draw[ thick](.562,.005)--(.945,.005);
\draw[ domain=0.15:.5629,smooth,variable=\x, thick] plot ({\x},{-1.3*(\x-.255)+.4});


\draw[blue, domain=0.03:1,smooth,variable=\x, thick] plot ({\x},{(1-\x)^1.65});
\draw[thick, red](.03, .945)--(.3,0)--(1,0);

\draw[dotted] (0.1,0.185)node[left]{$c_{G^\varnothing}(\underline r)$}--(.255,.185);
\draw[dotted] (0.1,0.4)node[left]{$c_{G'} (\underline r)$}--(.255,.4);

\draw [decorate, decoration={brace, amplitude=4pt}](0.255,0.4) -- (0.255,.185) node [black,midway,xshift=0.4 cm, yshift=-.1cm]
{\footnotesize $p'$};


\draw [dotted] (.25,-.02) node[below]{$\underline r$}--(.25,.615);\end{tikzpicture}
\caption{Arbitrary binary distribution}
\label{fig:44a}
\end{subfigure} 
\hspace*{8em}
\captionsetup[subfigure]{oneside,margin={5cm,-10cm}}\begin{subfigure}[t]{0.3\textwidth}
\centering
\begin{tikzpicture}[xscale=8, yscale=4.8]
\draw [<->, help lines, ultra thick] (0,0) -- (1.1,0);
\draw [->, help lines, ultra thick] (0.1,0) --(.1,1.1);\node[right=3pt] at (1.08,0) {$r$};
\node[left=3pt] at (0.1, 1.1) {$c$};

\draw [] (1,-.02)node[below]{$\bar\theta$}--(1,.02);
\draw [] (.03,-.02) node[below]{$\underline\theta$}--(.03,.02);
\draw [dotted] (0.1, .945)node[right]{$m_\varnothing -\underline\theta$}--(0.03,.945);
\draw [dotted] (.03, 0)--(.03,.945);

\draw[ Green, thick](.0327, .947)--(.082,.775);
\draw[Green, thick](.67,.008)--(.945,.008);
\draw[Green, domain=0.077:.674,smooth,variable=\x, thick] plot ({\x},{-1.3*(\x-.255)+.55});

\draw[blue, domain=0.03:1,smooth,variable=\x, thick] plot ({\x},{(1-\x)^1.65});
\draw[thick, red](.03, .945)--(.3,0)--(1,0);

\draw[dotted] (0.1,0.185)node[left]{$c_{G^\varnothing}(\underline r)$}--(.255,.185);
\draw[dotted] (0.1,0.558)node[left]{$ c_{G''}(\underline r)$}--(.257,.558);

\draw [decorate,decoration={brace,amplitude=6pt}]
(0.255,0.542) -- (0.255,.185) node [black,midway,xshift=0.45cm]
{\footnotesize $p''$};

\draw [dotted] (.25,-.02) node[below]{$\underline r$}--(.25,.553);
\end{tikzpicture}
\vspace*{-7mm}\caption{Profit-improving binary distribution}
\label{fig:44b}
\end{subfigure}
 \captionsetup{justification=justified,singlelinecheck=false}
\caption{In both panels, the red curve is~$c_{G^\varnothing}$ and the blue curve is $c_F$. Panel (a) depicts $p'=c_{G'}(\underline r)-c_{G_\varnothing}(\underline r)$ and function $c_{G'}$ in black. Panel (b) depicts $p''=c_{G''}(\underline r)-c_{G_\varnothing}(\underline r)$ and function $c_{G''}$ in~green. $c_{G''}$ has the same slope as $c_{G'}$ at point $\underline r$ but $G'$ is a mean-preserving contraction of $G''$, which is depicted by the fact that $c_{G'}\leq c_{G''}$ pointwise. }
\label{fig:44}
\end{center}
\end{figure}

Therefore, the principal's objective in \eqref{eq:pm} is satisfied by finding a piece-wise linear function $c_G$ that, at $\underline r$, optimally trades off being as large and as flat as possible, this is, between higher surplus extraction and a higher probability that the agent continues his search. Notice that the candidate binary distribution $G'$ delivers a lower profit than the binary distribution $G''$, where $c_{G''}$ is depicted by the green piece-wise linear convex curve in \autoref{fig:44b}:  at $\underline r$, both $c_{G''}$ and $c_{G'}$ have the same slope but the former has a strictly higher value, which implies that the principal persuades the agent to continue searching with the same probability by offering either $G'$ or $G''$  but she charges a strictly higher price for $G''$. Consequently, $c_G$ must be tangent to $c_F$ at some point of $\Theta$. Formally, there is some $x\in(\underline \theta,\bar \theta)$ such that $G=G^{\pi(x)}$--which, recall, is induced by a ``pass/fail" signal that sends a ``fail" signal realization for any $\theta \leq x$ and a ``pass" realization for any $\theta> x$. The agent continues searching when he sees a ``fail" if $\mathbb{E}_F[\theta|\theta\leq  x]\leq \underline r$  and terminates his search when he sees a ``pass" if $\mathbb{E}_F[\theta|\theta>  x]>\underline r$. 

Of course, in a stationary equilibrium, the value of the principal's profit maximization problem \eqref{eq:pm} must also be consistent with the self-generation condition \eqref{eq:sgp}. We show that, once we incorporate the self-generation condition in the profit maximization problem, and restrict attention to the family of ``pass/fail" signals $\{G^{\pi(x)}\}_{x\in (\underline \theta, \bar \theta)}$ as defined in \eqref{eq:3}, the principal's infinite-dimensional profit-maximization problem reduces to a \emph{single-dimensional constrained surplus maximization} problem:
\begin{align*}
V=&\max_{x\in(\underline \theta, \bar \theta)} \hspace*{.2em} u(G^{\pi(x)})-\underline u \\[6pt]
\label{eq:fail}
\tag{fail-IC}
\text{s.t. } & \mathbb{E}_F[\theta|\theta\leq  x]\leq \underline r\\[6pt]
\label{eq:pass}
\tag{pass-IC}
&\mathbb{E}_F[\theta|\theta>  x]>\underline r .
\end{align*}

A natural candidate solution is $x=\bar r$, which induces the agent to stop her search if and only if a good's quality exceeds the McCall reservation value. If we disregard the incentive compatibility constraints, it is clear that this point of tangency generates the highest possible surplus of $\bar u-\underline u$. 

We conclude the sketch by arguing that the candidate solution of $x=\bar r$ indeed satisfies \eqref{eq:fail} and \eqref{eq:pass}. To see why the two constraints are satisfied, consider \autoref{fig:optimal_pass_fail} which depicts $c_{G^{\pi(\bar r)}}$ as the black piece-wise convex curve that is tangent to $c_F$ at $x=\bar r$. The kinks occur at the support points $m_1^*=\mathbb{E}_F[\theta|\theta\leq \bar r]$ and $m_2^*=\mathbb{E}_F[\theta|\theta>\bar r]$. It is clear from the figure that $m_1^*<\underline r<m_2^*$, implying that \eqref{eq:fail} and \eqref{eq:pass} are both slack. Indeed, this observation is not a coincidence but a general property; we prove that if the IC constraints are not slack at $x=\bar r$, then the agent would have zero willingness-to-pay for any signal the principal could offer, which cannot occur in non-trivial settings that satisfy \hyperref[ass:1]{\Cref{ass:1}}. Therefore, $x=\bar r$ is the solution to the constrained surplus maximization problem, and  $V=\bar u-\underline u$.

\bigskip

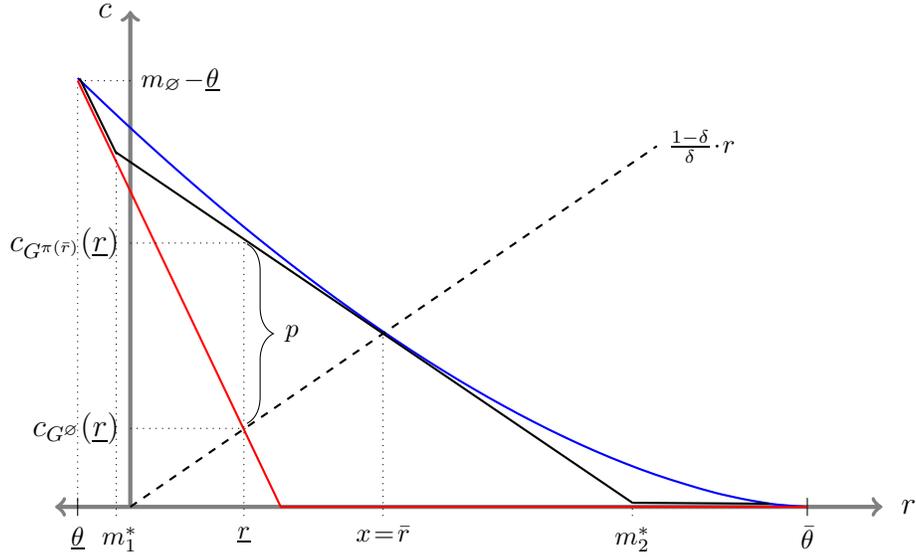
\begin{figure}[ht]
\begin{center}
\begin{tikzpicture}[xscale=10, yscale=6]
\draw [<->, help lines, ultra thick] (0,0) -- (1.1,0);
\draw [->, help lines, ultra thick] (0.1,0) --(.1,1.1);\node[right=3pt] at (1.1,0) {$r$};
\node[left=3pt] at (0.1, 1.1) {$c$};

\draw [] (1,-.02)node[below]{\footnotesize{$\bar\theta$}}--(1,.02);
\draw [] (.03,-.02) node[below]{\footnotesize{$\underline\theta$}}--(.03,.02);
\draw [dotted] (0.1, .945)node[right]{\footnotesize{$m_\varnothing -\underline\theta$}}--(0.03,.945);
\draw [dotted] (.03, 0)--(.03,.945);

\draw[blue, domain=0.03:1,smooth,variable=\x, thick] plot ({\x},{(1-\x)^1.65});
\draw[thick, red](.03, .945)--(.3,0)--(1,0);
\draw[domain=0.0799:.768,smooth,variable=\x, thick] plot ({\x},{-1.131895*(\x-.44)+.3798});
\draw[ thick](.0327, .95)--(.0813,.784);
\draw[ thick](.7675,.009)--(.95,.007);

\draw[dotted](.0813,-0.02)node[below]{\footnotesize{~$m^*_1$}}--(.0813,.784);
\draw[ dotted](.7675,-0.02)node[below]{\footnotesize{$m^*_2$}}--(.7675,0.0);
\draw[thick, dashed] (0.1,0)--(.8,.8)node[right]{\footnotesize{$\frac{1-\delta}{\delta}\cdot r$}};

\draw [dotted] (.251,-.02) node[below]{\footnotesize{$\underline r$}}--(.251,.59);
\draw [dotted] (.436,-.02) node[below]{\footnotesize{$x=\bar r$}}--(.436,.39);

\draw [decorate,decoration={brace,amplitude=10pt}]
 (0.255,0.583) -- (0.255,.185) node [black,midway,xshift=0.6cm]
 {\footnotesize $p$};

 \draw[dotted] (0.1,0.174)node[left]{$c_{G^\varnothing}(\underline r)$}--(.255,.174);
\draw[dotted] (0.1,0.585)node[left]{$c_{G^{\pi(\bar r)}} (\underline r)$}--(.255,.585);





\end{tikzpicture}
\captionsetup{justification=justified,singlelinecheck=false}
\caption{The red curve is $c_{G^\varnothing}$, the blue curve is $c_F$, and the black curve is $c_{G^{\pi(\bar r)}}$. Note that the support points satisfy $m_1^*<\underline r<\bar r<m_2^*$.}
\label{fig:optimal_pass_fail}
\end{center}
\end{figure}

We now make a number of observations regarding the findings of \autoref{prop:1}. First, while the price and the continuation values in a stationary equilibrium are uniquely pinned down, there are  an uncountable number of posterior-mean distributions that can support a stationary equilibrium. In contrast to the class of pass/fail signals we highlight in our proof sketch, consider a lower-censorship signal that  reveals a good's quality if it exceeds $\bar r$ and pools all goods whose quality lies below $\bar r$ by labeling them as ``fail." This signal induces the distribution $G^{LC}$ given by 
\[
G^{LC}(m)=\left\{\begin{array}{ccl}
   0 &\mbox{if} &~~~ m<\mathbb{E}_F[\theta|\theta\leq \bar   r]  \\
       F(\bar r) &\mbox{if} &~~~ \mathbb{E}_F[\theta|\theta\leq \bar r]\leq m<  \bar r  \\
     F(m) &\mbox{if} & ~~~m\geq \bar r
     \end{array}\right. .
\]
Note that this lower-censorship signal is Blackwell more informative than the pass/fail signal that induces $G^{\pi(\bar r)}$. Hence, $G^{\pi(\bar r)}$ is a mean-preserving contraction of $G^{LC}$---implying $c_{G^{\pi(\bar r)}}\leq c_{G^{LC}}$ pointwise as can be seen in \autoref{fig:optdist}---and $G^{LC}$ also supports a stationary equilibrium as it satisfies points $(iv)$ and $(v)$ of \autoref{prop:1}. Despite the multitude of equilibrium distributions, all of them take the same value on $[\underline r,  \bar r]$. Thus, all signals supporting a stationary equilibrium induce the same search behavior in the agent, and hence the principal and agent are both indifferent across the set of equilibrium signals.
 
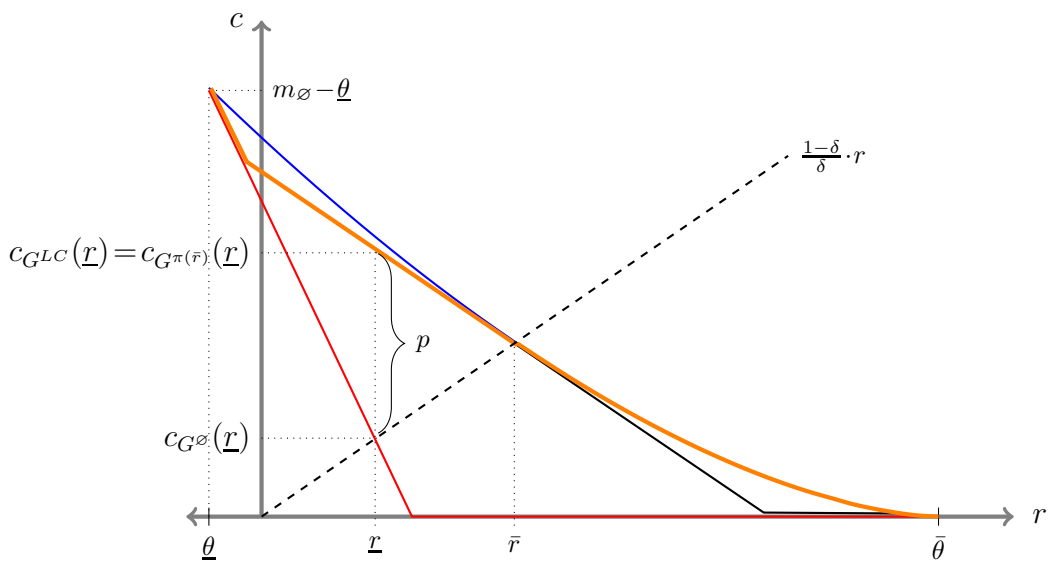
\begin{figure}[ht]
\begin{center}
\begin{tikzpicture}[xscale=10, yscale=6]
\draw [<->, help lines, ultra thick] (0,0) -- (1.1,0);
\draw [->, help lines, ultra thick] (0.1,0) --(.1,1.1);\node[right=3pt] at (1.1,0) {$r$};
\node[left=3pt] at (0.1, 1.1) {$c$};

\draw [] (1,-.02)node[below]{\footnotesize{$\bar\theta$}}--(1,.02);
\draw [] (.03,-.02) node[below]{\footnotesize{$\underline\theta$}}--(.03,.02);
\draw [dotted] (0.1, .945)node[right]{\footnotesize{$m_\varnothing -\underline\theta$}}--(0.03,.945);
\draw [dotted] (.03, 0)--(.03,.945);

\draw[blue, domain=0.03:1,smooth,variable=\x, thick] plot ({\x},{(1-\x)^1.65});
\draw[thick, red](.03, .945)--(.3,0)--(1,0);
\draw[opacity=.3, domain=0.0799:.768,smooth,variable=\x, thick] plot ({\x},{-1.131895*(\x-.44)+.3798});
\draw[opacity=.3, thick](.0327, .95)--(.0813,.784);
\draw[opacity=.3, thick](.7675,.009)--(.95,.007);

\draw[domain=0.0799:.44,smooth,variable=\x, orange, ultra thick] plot ({\x},{-1.131895*(\x-.44)+.3798});

\draw[ orange, ultra thick](.0327, .95)--(.0813,.784);
\draw[orange, domain=0.44:1,smooth,variable=\x, ultra thick] plot ({\x},{(1-\x)^1.65});

\draw[thick, dashed] (0.1,0)--(.8,.8)node[right]{\footnotesize{$\frac{1-\delta}{\delta}\cdot r$}};

\draw [dotted] (.251,-.02) node[below]{\footnotesize{$\underline r$}}--(.251,.59);
\draw [dotted] (.436,-.02) node[below]{\footnotesize{$\bar r$}}--(.436,.39);

\draw [decorate,decoration={brace,amplitude=10pt}]
 (0.255,0.583) -- (0.255,.185) node [black,midway,xshift=0.6cm]
 {\footnotesize $p$};

 \draw[dotted] (0.1,0.174)node[left]{$c_{G^\varnothing}(\underline r)$}--(.255,.174);
\draw[dotted] (0.1,0.585)node[left]{$c_{G^{LC}} (\underline r)=c_{G^{\pi(\bar r)}} (\underline r)$}--(.255,.585);





\end{tikzpicture}
\captionsetup{justification=justified,singlelinecheck=false}
\caption{The red curve is $c_{G^\varnothing}$, the blue curve is $c_F$, the orange curve is $c_{G^{LC}}$, and the light gray curve is $c_{G^\pi(\bar r)}$. Note that $c_{G^{\pi(\bar r)}}(x')= c_{G^{LC}}(x')$ for all  $x'\in [\underline \theta,\bar r]\cup\{\bar\theta\}$, and $c_{G^{\pi(\bar r)}}(x')<c_{G^{LC}}(x')$ for all $x'\in (\bar r,\bar \theta).$ }
\label{fig:optdist}
\end{center}
\end{figure}

\begin{corollary}
\label{cor:1}
If a contract $\langle p, G\rangle\in\mathbb{R}_+\times\mathcal{G}(F)$ supports a stationary equilibrium, then $G(m)=F(\bar r)$ for all  $m\in [\underline r, \bar r]$.
\end{corollary}

Second, the indeterminacy highlighted in point $(v)$ of \autoref{prop:1} potentially changes in the presence of costly information production, i.e., if the principal incurs differential costs for inducing different posterior-mean distributions. In particular, suppose the cost of inducing a distribution $G\in \mathcal{G}(F)$ is strictly increasing in the cardinality of its support, $|\supp(G)|$. In this case, the McCall ``pass/fail" signal is the unique equilibrium signal assuming the cost of providing this signal is no more  than $\bar u-\underline u$; otherwise, no information will be generated in equilibrium. Alternatively, 
if the cost of information is strictly increasing in the Blackwell order, then all distributions that support a stationary equilibrium are either binary distributions or uninformative (i.e., $G^\varnothing$).\footnote{ The cost of information is increasing in the Blackwell order if  whenever signal $(\pi, S)$ is Blackwell more informative than signal $(\pi', S')$, the former costs more than the latter.} Whether the McCall ``pass/fail" signal supports a stationary equilibrium for a given Blackwell monotone cost function depends on the additional trade-off between the marginal increase in surplus from inducing a more efficient search and the marginal information cost of doing so.\footnote{An alternative form of production costs from those described in this paragraph is as follows: Let $\widetilde F\in \mathcal{G}(F)$ represent the most informative posterior-mean distribution the principal can generate. Formally, any $G\in \mathcal{G}(\widetilde F)$ can be produced by the principal at zero cost, while producing any $G\notin \mathcal{G}(\widetilde F)$ is prohibitively costly. \autoref{prop:1} continues to hold after replacing $F$ with $\widetilde F$ and $\bar u=u(F)$ with $u(\tilde F)$ throughout the statement of the result.}

Third, since the agent's continuation and reservation values depend on $\delta$,  the discount factor affects how long the agent searches in an equilibrium, as described below:

\begin{corollary}\label{cor_delta_comp}
    The expected duration of search in any stationary equilibrium is non-decreasing in  $\delta.$ Additionally, the expected duration of search goes to infinity as $\delta$ goes to one.
\end{corollary}

\section{Public Signals}\label{public}
Thus far, we have assumed that the agent is  solely dependent on the principal to learn about a good's quality. However, there are settings in which public information is available about the sampled good. For example, a firm searching to fill a job vacancy may encounter candidates with different levels of education (or work history or criminal background), and those with different levels of education may have different distributions of ability.

In this section, we consider an agent who learns about a good's quality both by observing a signal $(\xi, Z)$, where $Z$ is a compact set of signal realizations and $\xi:\Theta\to \Delta (Z)$, and by contracting with the principal. We refer to $(\xi, Z)$ as a public signal because its realizations are observed by both the principal and the agent for free. Without loss of generality, we assume that $\xi$ has full support on $Z$. Additionally, for ease of exposition, we assume that the posterior beliefs conditional on observing any $z\in Z$ are atomless. 

We amend the timing of events within each period $t\geq 1$ as follows: First, the agent samples a good of quality $\theta_t$ which is unobserved by  the agent and the principal. However, both observe a public signal realization $z_t\in Z$ which is distributed according to $\xi(\theta_t)$. The principal then offers a  contract $\langle p_t, (\pi_t, S_t)\rangle$ which the agent either accepts or rejects. If the agent accepts the contract $(a_t=1)$, he pays $p_t$ and observes an additional signal realization $s_t\in S$ distributed according to $\pi_t(\theta_t)$. If the agent instead rejects the contract $(a_t=0)$, he pays nothing and observes no additional signal. Finally, the agent decides whether to consume the sampled good and stop searching $(d_t=1)$ or to continue searching $(d_t=0)$. \autoref{fig:timing2} depicts the amended timing of events. 

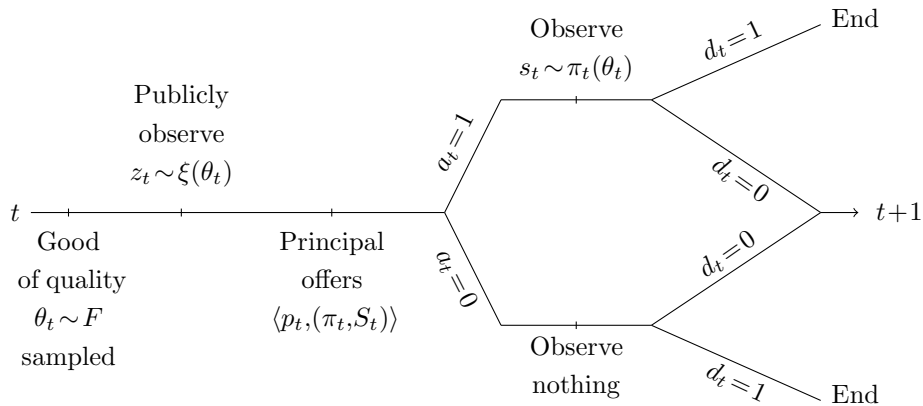
\begin{figure}[ht]
\centering 
\begin{tikzpicture}[scale=0.5]
\tikzset{
    position label/.style={
             text depth = 1ex
    }
}
    \foreach \x in {-8, -5, -1}
      \draw (\x cm,3pt) -- (\x cm,-3pt);
\draw[-] (-9,0) node[left]{\footnotesize{ $t$}}--(2,0);
\draw[-] (2,0)--(3.5,3);
\draw[-] (2,0)--(3.5,-3);
\draw[-] (3.5,-3)--(7.5,-3);
\draw[-] (3.5,3)--(7.5,3);
\draw[-] (7.5,-3)--(12,0);
\draw[-] (7.5,-3)--(12,-5);
\draw[-] (7.5,3)--(12,5);
\draw[-] (7.5,3)--(12,0);
\draw[->] (12,0)--(13,0)node[right]{\footnotesize{ ${t+1}$}};

\node [position label, align=center, below=3pt] (fend) at (-8,0) {\footnotesize{Good} \\ \footnotesize{ of quality}\\ \footnotesize{$\theta_t\sim F$}\\ \footnotesize{sampled}};
\node [position label, align=center, above=3pt] (fend) at (-5,0) {\footnotesize{Publicly} \\  \footnotesize{observe}\\ \footnotesize{$z_t\sim \xi(\theta_t)$}};
\node [position label, align=center, below=3pt] (fend) at (-1,0) {\footnotesize{Principal} \\ \footnotesize{offers}\\ \footnotesize{ $\langle p_t, (\pi_t, S_t)\rangle$}};
\node [position label, rotate=63, align=center, above] (fend) at (3,1.5) {\footnotesize{$a_t=1$}};
\node [position label, rotate=297, align=center, below] (fend) at (2.8,-1.5) {\footnotesize{$a_t=0$}};

\draw[](5.5,3.1)--(5.5,2.9);
\draw[](5.5,-3.1)--(5.5,-2.9);

\node [position label, align=center, above] (fend) at (5.5,3) {\footnotesize{Observe}\\ \footnotesize{$s_t\sim\pi_t(\theta_t)$}};
\node [position label, align=center, below] (fend) at (5.5,-3) {\footnotesize{Observe}\\ \footnotesize{nothing}};

\node [position label, align=center, right] (fend) at (12,5) {\footnotesize{End}};
\node [position label, align=center, right] (fend) at (12,-5) {\footnotesize{End}};

\node [position label, rotate=335, align=center, below] (fend) at (10,-4) {\footnotesize{$d_t=1$}};
\node [position label, rotate=35, align=center, above] (fend) at (10,-1.8) {\footnotesize{$d_t=0$}};

\node [position label, rotate=25, align=center, above] (fend) at (10,3.8) {\footnotesize{$d_t=1$}};
\node [position label, rotate=325, align=center, below] (fend) at (10.2,1.35) {\footnotesize{$d_t=0$}};
  \end{tikzpicture}
  \caption{Timing of events with a public signal}
  \label{fig:timing2}
\end{figure} 

Consider first the case of autarky in which the agent can only observe a realization of the public signal $(\xi, Z)$. His induced posterior-mean distribution is $G^\xi\in\mathcal{G}(F)$ and his expected payoff in the search market is $u(G^\xi)$, which is the solution to \eqref{eq:1'} evaluated at $G^\xi$. Note that $u(G^\xi)\geq \underline u$, i.e., the agent's autarky payoff improves in the current setting since he has access to a free signal. Therefore, the full surplus that can be generated by contracting with the principal is now $\bar u - u(G^\xi).$

Next, consider the case in which the agent can observe the public signal followed by a second signal. Given a public signal realization $z\in Z$, the agent updates his belief from his prior $F$ to an interim belief $F_{z}$. When the agent further updates his belief based on the second signal, $F_z$ serves as his new ``prior."\footnote{Mathematically, it is therefore equivalent to view this model as one in which the distribution of good quality is itself drawn from a distribution of priors with support $\{F_z\}_{z\in Z}$ in each period according to a distribution $\xi$.} Thus, the set of posterior-mean distributions that can be induced through a second signal is the set of mean-preserving contractions of $F_z$, i.e., $\mathcal{G}(F_z)$. For example, a fully informative second signal induces $F_z$ itself, whereas an uninformative  second signal induces the degenerate distribution
\[
G^\varnothing_z(m)=\left\{\begin{array}{ccl}
     0 &\mbox{if} & m<\mathbb{E}_{F_z}[\theta]  \\
     1 &\mbox{if} & m\geq \mathbb{E}_{F_z}[\theta]  
     \end{array}\right..
\]

Therefore, following a public signal realization $z\in Z$, the principal can propose any contract of the form $\langle p, G\rangle\in\mathbb{R}_+\times \mathcal{G}(F_z)$. Our goal is to characterize stationary equilibrium outcomes. However, we must first update the conditions  for a stationary equilibrium given in \autoref{def:1} to account for the presence of a public signal.  

\begin{definition}
    \label{def:2}
    \normalfont Given a public signal $(\xi, Z)$, a stationary equilibrium  is defined by a pair of  continuation values $(U, V)\in \mathbb{R}_+^2$ and a family of  contracts $\{\langle p_z,G_z\rangle\}_{z\in Z}$ with $\langle p_z, G_z\rangle\in\mathbb{R}_+\times \mathcal{G}(F_z)$ for each $z\in Z$ such that in each period and after every history:
    \begin{enumerate}
    \item The agent stops searching and consumes a good of expected quality $m\in \Theta$ if and only if $m>\delta U$, i.e.,  \eqref{eq:os} holds,  

    \item Given a public signal realization $z\in Z$, the agent accepts  $\langle \widehat p,\widehat G\rangle\in\mathbb{R}_+\times\mathcal{G}(F_z)$  if and only if
    \begin{align*}
          \label{eq:pc2}
    \tag{PC$_z$}
           c_{\widehat G}(\delta U)-c_{G_z^\varnothing}(\delta U) \geq \widehat p,
    \end{align*}
       
\item Given a public signal realization $z\in Z$, the principal proposes $\langle p_z,G_z\rangle$ to maximize her profit, i.e.,
\[
\label{eq:pm2}
\tag{PM$_z$}
 \langle p_z,G_z\rangle\hspace*{.2em}\in \argmax_{\langle \widehat p, \widehat G\rangle\in\mathbb{R}_+\times \mathcal{G}(F_z)}\hspace*{.2em} \widehat p+ \widehat G(\delta U) \delta V \hspace*{.5em} \text{ subject to \eqref{eq:pc2}},
\]   
\item Payoffs are self-generating, i.e.,
\[
\label{eq:sga2}
\tag{SG-A$^\xi$}
U=\int_Z\Big(\delta U+c_{G_z}(\delta U)-p_z\Big) \xi(dz)
\] 
    and
    \[
    \label{eq:sgp2}
\tag{SG-P$^\xi$}
V= \int_Z\Big(p_z + G_z(\delta U) \delta V\Big)\xi(dz).
\] 
\end{enumerate}
\end{definition} 

Compared to \autoref{def:1}, the sequential rationality conditions in \autoref{def:2} must hold following every public signal realization $z\in Z$, which is captured by the $z$-dependent participation constraint \eqref{eq:pc2} and profit maximization problem \eqref{eq:pm2}. In contrast, the self-generation conditions \eqref{eq:sga2} and \eqref{eq:sgp2} hold only in expectation.

In order to state our main result of this section, we must first introduce some notation. Let $r^\xi\triangleq r(G^\xi)$ be the agent's reservation value in autarky when his only source of information is the pubic signal. Given a public signal realization $z\in Z$, we denote the lower and upper bounds on the support of $F_z$ by $\underline \theta_z\triangleq\inf \hspace*{.2em} \supp(F_z)$ and $\bar \theta_z\triangleq\sup \hspace*{.3em}\supp(F_z)$ respectively. For  $k\geq 0$, let
\[
Z_A(k)\triangleq\{z\in Z:r^\xi \hspace*{.1em} < \hspace*{.1em}\underline \theta_z<\bar r-k\}.
\]
For each $z\in Z$ such that $r^\xi\in [\underline\theta_z, \bar\theta_z]$, let 
$\bar x_z\triangleq\sup\{x\in [\underline\theta_z, \bar\theta_z]:\mathbb{E}_{F_z}[\theta|\theta\leq x]\leq r^\xi\}$,\footnote{The mapping $x\mapsto \mathbb{E}_{F_z}[\theta|\theta\leq x]$ is  increasing, and (as $F_z$ is assumed to be atomless for all $z\in Z$) $\mathbb{E}_{F_z}[\theta|\theta\leq \underline\theta_z]=\underline \theta_z\leq r^\xi$. Thus, $\bar x_z$ is well-defined.} and for $k\geq 0$, let 
\[
Z_B(k)\triangleq\big\{z\in Z:r^\xi\in [\underline \theta_z ,\bar\theta_z]\hspace*{.1em} \wedge \hspace*{.1em}\bar x_z\hspace*{.1em} < \min\{\bar \theta_z, \hspace*{.1em}\bar r-k\}\big\}.
\]
Note that $Z_A(k)$ and $Z_B(k)$ could be empty. Finally, we define  the function $\Phi:\mathbb{R}_+\to \mathbb{R}$ by
       \small{\begin{align*}
        \label{eq:5}
        \tag{4}
\Phi(k)\triangleq \int_{\bar r-k}^{\bar r}F(m)dm+\int_{Z_A(k)}\int_{\underline\theta_z}^{\bar r-k}F_z(m)dm \xi(dz)+\int_{Z_B(k)}\int_{\bar x_z}^{\bar r -k}\big(F_z(m)-F_z(\bar x_z)\big)dm\xi(dz). 
        \end{align*}} 
        
       \normalsize 
\noindent We provide a detailed discussion of  $Z_A(\cdot)$, $Z_B(\cdot)$, and $\Phi(\cdot)$ following the statement of the following theorem.

\begin{theorem}
\label{prop:2}
    For any public signal $(\xi, Z)$, a stationary equilibrium exists. Furthermore, in any stationary equilibrium, 
    \begin{enumerate}[$(i)$]
        \item $U=u(G^\xi)$ and
        \item $V=\bar u-u(G^\xi)-\displaystyle\frac{k^*}{\delta}$, where $k^*\in [0, \bar r-r^\xi]$ is the unique solution to the fixed point 
        \[
        k=\delta\Phi(k).
        \]
    \end{enumerate}
\end{theorem}

\autoref{prop:2} establishes that for any given public signal, a stationary equilibrium exists and that all stationary equilibria are payoff equivalent. The agent gets his autarky payoff $u(G^\xi)$, which is now (potentially strictly) higher than $\underline u$, whereas the principal's payoff may be lower than the maximal surplus that she could feasibly extract $\bar u-u(G^\xi)$. In other words, introducing a free and public signal may lead to an inefficient outcome with $k^*/\delta$ measuring the total welfare loss in equilibrium.\footnote{The proof for \autoref{prop:2} establishes a ``full characterization'' similar to that in \autoref{prop:1}  by specifying for each public signal realization: the unique equilibrium price, the unique equilibrium stopping probability,  and  the set of equilibrium signals.}

To understand when the public signal causes an efficiency loss, suppose for some $k\in [0, \bar r-r^\xi]$, the principal seeks to induce a reservation value of $x_z=\bar r - k$ for all  $z\in Z$ by offering the pass/fail signal $\pi(\bar r-k)$. The loss of surplus stemming from the difference between the induced reservation value $\bar r-k$ and the McCall reservation value $\bar r$ is characterized by $\Phi(k).$  To see why, notice that the decrease in surplus from inducing $\bar r-k$ instead of $\bar r $ given a public signal realization $z\in Z$ is captured by 
\[
\mathbb{E}_{F_z}[\max\{\theta, \bar r\}]-\mathbb{E}_{F_z}[\max\{\theta, \bar r-k\}]=\int^{\bar r}_{\bar r-k}F_z(m)dm. 
\]
 which in expectation yields the first term in \eqref{eq:5} after noting that $\mathbb{E}_{\xi}[F_z]=F$ by Bayes-plausibility. 

However, the first term in \eqref{eq:5} captures the change in the total surplus assuming the principal can induce $\bar r-k$ for all realizations of the public signal. Yet, there are two cases when the principal cannot induce the desired reservation value following a public signal realization $z\in Z$. First, suppose the agent observes a public signal realization $z\in Z$ such that $r^\xi<\underline \theta_z$. In this case, no additional information can affect the agent's search as he already observes conclusive evidence that the quality of the sampled good is above his autarky reservation value, leading him to terminate his search. However, when $\underline\theta_z<\bar r-k$, the principal would like the agent to continue searching when the sampled good's quality lies in the interval $[\underline\theta_z, \bar r-k]$. This is precisely the case when $z\in Z_A(k)$, and the additional surplus loss is captured by 
 \[
\int_{\underline \theta_z}^{\bar r-k}\big(\bar r-k-m\big)dF_z(m)=\int_{\underline \theta_z}^{\bar r-k}F_z(m)dm,
 \]
 which in expectation yields the second term in \eqref{eq:5}.

Second, suppose the agent observes a public signal realization $z\in Z$ such that $r^\xi\in [\underline \theta_z ,\bar\theta_z]$. If $\bar x_z<\bar\theta_z$ (which is possible only if $\mathbb{E}_{F_z}[\theta]>r^\xi$), then the agent stops his search without additional information. Nonetheless, the agent is persuaded to  continue his search if he observes a ``fail" from a pass/fail signal $\pi(x)$ as long as $x\leq \bar x_z$. Thus, if $\bar x_z<\min\{\bar \theta_z, \bar r-k\}$, the pass/fail signal $\pi(\bar r-k)$ cannot persuade the agent to continue his search even when the agent observes a ``fail" signal realization. This is precisely the case when $z\in Z_B(k)$. Intuitively, a public signal realization $z\in Z_B(k)$ induces an interim posterior belief $F_z$ that places more mass on the interval $[r^\xi, \bar r-k]$ than it does on $[\underline \theta_z,r^\xi]$, as shown in \autoref{fig:ICbinds}. Thus, even though the agent lacks conclusive evidence, he is more confident than not that the sampled good's quality is high enough for him to stop his search, making it difficult for the principal to convince him otherwise.

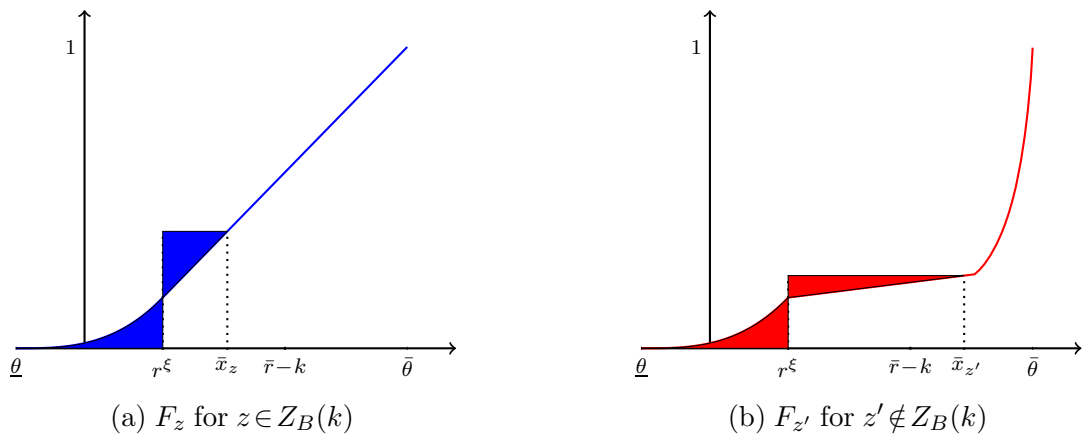
\begin{figure}[ht]
\begin{center}
\begin{subfigure}[t]{0.4\textwidth}
\begin{tikzpicture}[xscale=1.3, yscale=1]\scriptsize
\draw[->, thick] (0,0) -- (4.5,0);
\draw[->, thick] (0.7,0) -- (0.7,4.5) ;
\draw [] (4,-.02)node[below]{$\bar\theta$}--(4,.02);
\draw [] (0,-.02)node[below]{$\underline\theta$}--(0,.02);
\draw [dotted, thick] (1.5,-.02)node[below]{$r^\xi$}--(1.5,1.556);

\draw [dotted, thick] (2.16,-.02)node[below]{$\bar x_z$}--(2.16,1.556);

\draw [thick] (2.75,-.02)node[below]{$\bar r-k$}--(2.75,.02);

\draw (0.7,4)node[anchor=east]{1};
\draw[blue, domain=0:1.5,smooth,variable=\x, thick] plot ({\x},{.2*(\x)^3});

\draw[blue, domain=1.497:4,smooth,variable=\x, thick] plot ({\x},{1.3333*(\x)-1.324});

\draw[fill=blue,opacity=.2, domain=0:1.5,smooth,variable=\x] plot ({\x},{.2*(\x)^3})--(1.5,0)--(0,0);

\draw[fill=blue,opacity=.2, domain=1.497:2.16,smooth,variable=\x] plot ({\x},{1.3333*(\x)-1.324})--(1.5,1.556)--(1.5,.675);

\end{tikzpicture}	
\caption{$F_z$ for $z\in Z_B(k)$}
\label{fig:IC1}
\end{subfigure}
\hspace{20mm}
\begin{subfigure}[t]{0.4\textwidth}
\begin{tikzpicture}[xscale=1.3, yscale=1]\scriptsize
\draw[->, thick] (0,0) -- (4.5,0);
\draw[->, thick] (0.7,0) -- (0.7,4.5) ;
\draw [] (4,-.02)node[below]{$\bar\theta$}--(4,.02);
\draw [] (0,-.02)node[below]{$\underline\theta$}--(0,.02);
\draw [dotted, thick] (1.5,-.02)node[below]{$r^\xi$}--(1.5,.96755);

\draw [dotted, thick] (3.3,-.02)node[below]{~$\bar x_{z'}$}--(3.3,.96755);

\draw [thick] (2.75,-.02)node[below]{$\bar r-k$}--(2.75,.02);

\draw (0.7,4)node[anchor=east]{1};
\draw[red, domain=0:1.5,smooth,variable=\x, thick] plot ({\x},{.2*(\x)^3});

\draw[red, domain=1.497:3.4,smooth,variable=\x, thick] plot ({\x},{.1625*(\x)+.4313});

\draw[red,thick,-] plot [smooth, tension=1] coordinates{(3.4,.98) (3.8, 2) (4, 4)};

\draw[fill=red,opacity=.2, domain=0:1.5,smooth,variable=\x] plot ({\x},{.2*(\x)^3})--(1.5,0)--(0,0);

\draw[fill=red,opacity=.2, domain=1.497:3.3,smooth,variable=\x] plot ({\x},{.1625*(\x)+.4313})--(1.5,.96755)--(1.5,.675);

\end{tikzpicture}	
\caption{$F_{z'}$ for $z'\notin Z_B(k)$}
\label{fig:ICb}
\end{subfigure}
\captionsetup{singlelinecheck=false, justification=justified}
\caption{The two panels show the agent's interim posterior beliefs $F_z$ (in blue) and $F_{z'}$ (in red) after observing public signal realizations $z$ and $z'$, respectively. For ease of comparison, both $F_z$ and $F_{z'}$ have full support on $\Theta$ and $F_z(\theta)=F_{z'}(\theta)$ for all $\theta\leq r^\xi$. The point $\bar x_z$ is determined by the unique value such that the shaded area to the left of $r^\xi$  equals the shaded area to its right, i.e., $\int^{r^\xi}_{\underline \theta}F_z(\theta)d\theta=\int_{r^\xi}^{\bar x_z}[F_z(\bar x_z)-F_z(\theta)]d\theta$, and $\bar x_{z'}$ is defined~analogously.~As~$\bar x_z<\bar r-k<\bar x_{z'}$,~we~conclude that $z\in Z_B(k)$ and $z'\notin Z_B(k)$. }
\label{fig:ICbinds}
\end{center}
\end{figure}

When $z\in Z_B(k)$, the principal must make a concession by offering a pass/fail signal $\pi(x_z)$ with  $x_z\in[r^\xi, \bar x_z]$ in order to persuade the agent. For any two cutoffs  $x_z'<x_z''$ in $[r^\xi, \bar x_z]$, the principal can charge more for $\pi(x_z')$ (as $x_z'$ is closer to the agent's actual reservation value of $r^\xi$) whereas she is more likely to keep the agent in the market by offering $\pi(x_z'')$ (as more goods are failed). Therefore, there is a familiar intra-temporal trade-off between surplus extraction (a higher per-period price) and persuasion (a higher probability of keeping the agent searching). Our proof of \autoref{prop:2} shows that the principal fully prioritizes persuasion when $z\in Z_B(k)$; she offers the pass/fail signal $\pi(\bar x_z)$, which makes the agent indifferent between stopping and continuing his search conditional on seeing a ``fail." As a result, the highest price the principal can charge to induce the agent to buy this signal is zero. The additional surplus loss in this case is captured by 
  \[
\int_{\bar x_z}^{\bar r-k}\big(\bar r-k-m\big)dF_z(m)=\int_{\bar x_z}^{\bar r-k}\big(F_z(m)-F_z(\bar x_z)\big)dm,
 \]
 which in expectation yields the last term in \eqref{eq:5}.

The principal is able to induce the McCall reservation value in a stationary equilibrium following any public signal realization if and only if $Z_A(k)\cup Z_B(k)$ has $\xi$-measure zero for all $k\in [0, \bar r-r^\xi]$. To see this, note that by construction, both $Z_A(k)$ and $Z_B(k)$ shrink (decrease in the inclusion order) as $k$ increases. Therefore, if $\xi\left(Z_A(k)\cup Z_B(k)\right)>0$ for some $k\in [0, \bar r-r^\xi]$, then $\xi\left(Z_A(k')\cup Z_B(k')\right)>0$ for all $k'\in[0,k]$. Thus, inefficiencies arise in a stationary equilibrium if and only if $\xi\left(Z_A(0)\cup Z_B(0)\right)>0$, and the efficiency loss is quantified by the unique solution to the fixed point $k=\delta\Phi(k)$.
 
Before we conclude this section, we make three points. First, when $F_z$ has full support over $\Theta$ for all realizations $z\in Z,$ the set $Z_A(0)$ is trivially empty. The following result provides a sufficient and necessary condition for $Z_B(0)$ to also be empty in this case.

\begin{corollary}\label{cor:2}
Consider a public signal $(\xi, Z)$ such that $F_z$ has full support on $\Theta$ for almost all $z\in Z$. The principal extracts the full surplus in a stationary equilibrium if and only if $\mathbb{E}_{F_z}[\theta|\theta\leq \bar r]\leq r^\xi$ for almost all $z\in Z$.
\end{corollary}

Intuitively, the requirement that $\mathbb{E}_{F_z}[\theta|\theta\leq \bar r]\leq r^\xi$ for almost all $z\in Z$ is equivalent to insisting that the interim posterior belief induced by almost each public signal realization has a relatively ``fat left tail," that is, the posterior belief places more mass below $r^\xi$ than it does on $[r^\xi, \bar r]$. This condition is trivially satisfied when the public signal is uninformative as $\mathbb{E}_F[\theta|\theta\leq \bar r]\leq \underline r\leq r(G)$ for any $G\in\mathcal{G}(F)$ and any absolutely continuous prior $F$. Hence, our result in \autoref{prop:1} can be seen as a special case of \autoref{prop:2}.

Second, the presence of a public signal not only improves the agent's autarky payoff but it also shortens the expected duration of his search in equilibrium, as stated below:\footnote{While the presence of public signal weakly shortens the agent's expected duration of search when compared to the case without a public signal, it is not true that a (Blackwell) more informative public signal leads to a shorter expected search duration than a less informative public signal. To see why, notice that a perfectly informative public signal and a perfectly uninformative public signal both lead to the same expected search duration as in our baseline model, while \autoref{cor_info_comp} states that a partially informative pubic signal that yields an inefficient stationary equilibrium outcome must lead to a strictly shorter duration of search.}

\begin{corollary}\label{cor_info_comp}
  The expected duration of search in any stationary equilibrium is weakly shorter in the presence of a public signal, compared to without a public signal. Furthermore, the expected duration of search is strictly shorter in the presence of a public signal if and only if $k^*>0$, where $k^*$ is as defined in \autoref{prop:2}.
\end{corollary}

Finally,  in some markets the agent may also observe a private signal about the quality of each good before deciding to contract with the principal. The agent's willingness-to-pay for an additional signal would then be his private information, introducing a screening problem \citep[as in][]{ber18} into each stage game between the principal and the agent. Such an analysis would require a different set of tools; for example, truth-telling constraints for a privately-informed agent would be more tractable if formulated via recommendation mechanisms \citep[as in][]{kol17} rather than via distributions of posterior means as in our analysis. We therefore view the study of private signals in search markets with endogenously-supplied information as a fruitful avenue for future research.

\section{Conclusion}\label{conclude}

In this paper, we introduce a persuaded search game in which a profit-maximizing principal with limited commitment controls the supply and pricing of information while an agent decides whether to acquire information from the principal and when to stop searching. By decentralizing information and decision-making in an otherwise workhorse model of sequential search without recall, our model highlights intra-temporal trade-offs between surplus extraction and persuasion as well as how these trade-offs are inter-temporally linked due to the dynamic considerations for both the principal and agent.

When the principal is the sole source of information, we show that the essentially unique stationary equilibrium  is efficient and features full surplus extraction by the principal. This also implies that the principal does not benefit from offering more complex contracts that condition prices on the signal realizations and agent's decision to stop or continue searching, considering non-stationary strategies, or having long-term commitment power.

In contrast, when there are public and free sources of additional information, we show that the essentially unique stationary equilibrium  could be inefficient, but only if the public information induces interim beliefs that are either too ``optimistic" or have a sufficiently ``thin left tail." Perhaps counterintuitively, this implies that while a free and public source of information never lowers the agent's ex-ante payoff, the agent may nevertheless end up with an ex-post lower-quality match. This finding matches the empirical results of \cite{hof18} who show that firms which use pre-employment tests along with additional sources of (public) information such as resume screening  end up hiring less productive employees than firms which rely only on pre-employment tests. While they attribute this result to human error or bias among hiring managers, our model predicts this as an equilibrium outcome of a contracting game between rational players.

The parsimonious and tractable nature of our model is useful to derive additional results under different settings. In the Online Appendix, we show that \autoref{prop:1} qualitatively extends to a setting in which the agent incurs a flow cost from searching. We also consider an alternative setting in which the principal and the agent have heterogeneous discount factors. We show that the essentially unique stationary equilibrium involves an inefficiently long search with ex-post higher-quality matches when the principal is more patient than the agent, whereas it involves an inefficiently short search with ex-post lower-quality matches when the principal is less patient than the agent.

Of course, some search markets will contain different salient features that are not present in our model. For example, the payment scheme in real-estate market typically differs: the agent pays only upon successful termination of search. We study such a model in the Online Appendix and find that the principal's incentives flip---the principal now  persuades the agent to terminate his search, in contrast to our baseline model in which she persuades the agent to continue searching. However, we show that the principal still extracts the full surplus in a stationary equilibrium, mirroring our finding in \autoref{prop:1}. Of course, there are other payments schemes and market features that we do not study here---for example, our model does not address what happens if the underlying distribution of quality is non-stationary or if it is endogenously determined in equilibrium. The features of equilibrium contracts in these alternative models are, to our knowledge, unknown and studying such models may yield novel economic insights.

\singlespacing
 \bibliographystyle{abbrvnat}
\nocite{}\bibliography{bibref}
\onehalfspacing

\section*{Appendix}\label{appendix}

\noindent\begin{proof}[Proof of \autoref{lemma:1}]~

\noindent (\emph{Only-if direction}): Suppose the agent never searches, i.e., $\underline u=\bar u$. Then $\bar u=m_\varnothing$ because $\underline u=m_\varnothing$. Suppose for contradiction that $\underline \theta<\delta m_\varnothing$. Then 
\begin{align*}
m_\varnothing=\bar u=\int_\Theta \max\{m,\delta m_\varnothing\}dF(m)>\int_\Theta m dF(m) =m_\varnothing,
\end{align*}
where the second equality follows from \eqref{eq:1}, and the strict inequality follows from $F$ having positive mass over the range $[\underline \theta, \delta m_\varnothing]$, and the final equality follows from the definition of $m_\varnothing$. This however yields the contradiction $m_\varnothing>m_\varnothing$. Hence, if the agent never searches then $\underline \theta\geq \delta m_\varnothing$.\\

\noindent (\emph{If direction}):  Suppose $\underline\theta\geq \delta m_\varnothing$. Then for any $G\in\mathcal{G}(F)$, 
\begin{align*}
\int_\Theta\max\{m, \delta m_\varnothing\} dG(m)=\int_\Theta mdG(m)=m_\varnothing.
\end{align*}
Hence, $m_\varnothing$ is the solution to the fixed point problem in \eqref{eq:1}, i.e., $u(G)=m_\varnothing$ for all $G\in\mathcal{G}(F)$. In particular, $\bar u\triangleq u(F)=u(G^\varnothing)\triangleq\underline u$, so the agent never searches. \end{proof}\\

\noindent\begin{proof}[Proof of \autoref{prop:1}]
\subsection*{Only-if direction}
We first provide a proof of the necessary conditions: a tuple $(p, G, U, V)$ constitutes a stationary equilibrium only if it satisfies points $(i)$-$(v)$ of the theorem. 

To that end, consider a contract $\langle p, G\rangle $ along with a pair of agent-principal continuation values $(U, V)$ that constitute a stationary equilibrium. Then the contract $\langle p,G\rangle$ must  solve the constrained profit maximization problem \eqref{eq:pm}. Since the objective function in \eqref{eq:pm} is increasing in the price $\widehat p$, the constraint \eqref{eq:pc} must bind at the maximum, i.e., $p=c_{G}(\delta U)-c_{G^\varnothing}(\delta U)$. Additionally, the contract $\langle p, G\rangle $ must be consistent with the agent's self-generation condition \eqref{eq:sga} which, given the binding participation constraint, becomes   
\[
U=\delta U+c_{G^\varnothing}(\delta U).
\]
The above fixed-point problem is the same as \eqref{eq:1'} evaluated at $G^\varnothing$, and thus has a unique solution $U=u(G^\varnothing)\triangleq \underline u$,  establishing point $(i)$ of the theorem.

Since the agent earns his autarky payoff in a stationary equilibrium, from \eqref{eq:os}, he stops his search if and only if the expected quality of a good exceeds the aurarky reservation value $\underline r$. We can therefore rewrite the principal's profit maximization problem \eqref{eq:pm} as
\begin{align*}
\label{eq:4}
\tag{PM$'$}
&\max_{\widehat G\in\mathcal{G}(F)} \hspace*{.2em} \underbrace{c_{\widehat G}(\underline r)-c_{G^\varnothing}(\underline r)}_{=\widehat p}+\delta \widehat G(\underline r)V.
\end{align*}

In order to solve the optimization problem in \eqref{eq:4}, we first note that the agent has only two actions upon observing a signal realization\textemdash continue searching or stop. Thus, there exists a solution to \eqref{eq:4} that is a binary posterior-mean distribution, as stated in the lemma below. 

\begin{lemma}\label{lemma:binary}
A distribution $G\in\mathcal{G}(F)$ is a solution to  \eqref{eq:4} if and only if there exists a binary distribution $G^\dagger\in\mathcal{G}(F)$ with $\supp(G^\dagger)=\{m_1^\dagger, m_2^\dagger\}$ such that 
\begin{enumerate}[$(i)$]
    \item $m_1^\dagger\leq \underline r<m_2^\dagger$,
    \item $G^\dagger(\underline r)= G(\underline r)$,
   \item $c_{G^\dagger}(x)\leq c_G(x)$ for all $x\in\Theta$ with equality at $x=\underline r$, and
    \item $G^\dagger$ is a solution to \eqref{eq:4}.
\end{enumerate}
\end{lemma}

\noindent\begin{proof}[Proof of \autoref{lemma:binary}]
\noindent (\emph{Only-if direction}): Suppose $G\in\mathcal{G}(F)$ is a solution to \eqref{eq:4}. We construct a binary distribution satisfying points $(i)-(iv)$ of the lemma.

 We first argue that the agent must  both stop and continue his search in each period with a positive probability, i.e., $G(\underline r)\in (0,1)$. To that end, for any $V\geq 0$, we have 
\begin{align*}
    c_G(\underline r)-c_{G^\varnothing}(\underline r)+G(\underline r)\delta V \geq c_F(\underline r)-c_{G^\varnothing}(\underline r)+F(\underline r)\delta V>0
\end{align*}
where the first inequality is because $G$ solves \eqref{eq:4} and the last equality follows by \hyperref[ass:1]{Assumption 1}. 

Suppose, for the sake of a contradiction, that $G(\underline r)=0$, then 
\begin{align*}
    c_G(\underline r)&=c_G(\underline \theta)-\int_{\underline \theta}^{\underline r}\big(1-G(m)\big)dm \leq c_{G^\varnothing}(\underline \theta)-\int_{\underline \theta}^{\underline r}\big(1-G^{\varnothing}(m)\big)dm =c_{G^\varnothing}(\underline r),
\end{align*}
where the inequality follows because $c_G(\underline \theta)=c_{G^\varnothing}(\underline \theta)$ by point $(a)$ of  \autoref{remark:1} and because $0=G(m)\leq G^{\varnothing}(m)$ for all $m\leq \underline r$. Since $G^\varnothing$ is a mean-preserving contraction of $G$, we also have $c_G\geq c_{G^\varnothing}$ pointwise, which implies that $c_G(\underline r)=c_{G^\varnothing}(\underline r)$. However, this implies that $c_G(\underline r)-c_{G^\varnothing}(\underline r)+G(\underline r)\delta V=0$, which is a contradiction.

If instead $G(\underline r)=1$, then 
\[
c_G(\underline r)=\int^{\bar\theta}_{\underline r}\big(1-G(m)\big)dm=0
\]
while $c_{G^\varnothing}(\underline r)=m_\varnothing-\underline r=(1-\delta)m_\varnothing>0$. However, this contradicts the fact that $c_G\geq c_{G^\varnothing}$ pointwise. Therefore, we conclude that $G(\underline r)\in (0,1)$.

Consider the following binary distribution $G^\dagger$, which we will show satisfies points $(i)-(iv)$ of the lemma: 
\[
G^\dagger(m)=\left\{\begin{array}{ccl}
     0 &\mbox{if} &~~~ m<\mathbb{E}_{ G}[\theta|\theta\leq  \underline   r]  \\
        G(\underline r) &\mbox{if} &~~~ \mathbb{E}_{ G}[\theta|\theta\leq \underline r]\leq m<\mathbb{E}_{ G}[\theta|\theta> \underline r]  \\
     1 &\mbox{if} &~~~ m\geq \mathbb{E}_{ G}[\theta|\theta>  \underline r]
     \end{array}\right..
\]
The terms $\mathbb{E}_{ G}[\theta|\theta\leq  \underline   r]\triangleq m_1^\dagger$ and $\mathbb{E}_{ G}[\theta|\theta>  \underline r]\triangleq m_2^\dagger$ exist as $G(\underline r)\in (0,1)$, so $G^\dagger$ is well-defined. By construction, $m_1^\dagger\leq \underline r<m_2^\dagger$, which establishes both point $(i)$ of the lemma, and implies $G^\dagger(\underline r)=G(\underline r)$, establishing point $(ii)$ of the lemma. 

It is straightforward to see that $G^\dagger$ is constructed as a mean-preserving contraction of $G$, and hence $c_{G^\varnothing}\leq c_{G^\dagger}\leq c_G$ pointwise. Additionally, notice that 
\begin{align*}
    c_G(\underline r)&=\int_{\underline r}^{\bar \theta}(m-\underline r)dG(m)\\[6pt]
    &=\big(\mathbb{E}_{G}[\theta|\theta\geq \underline r]-\underline r\big)\big(1-G(\underline r)\big)\\[6pt]
    &=\big(\mathbb{E}_{G}[\theta|\theta\geq \underline r]-\underline r\big)\big(1-G^\dagger(\underline r)\big)\\[6pt]
&=c_{G^\dagger}(\underline r).
\end{align*}
Therefore, we have established point $(iii)$ of the lemma.
We can now conclude that
\[
    c_G(\underline r)-c_{G^\varnothing}(\underline r)+G(\underline r)\delta V = c_{G^\dagger}(\underline r)-c_{G^\varnothing}(\underline r)+G^\dagger(\underline r)\delta V,
\]
so $G^\dagger$ is also a solution to \eqref{eq:4}, which establishes point $(iv)$ of the lemma, and completes the proof of existence of a binary distribution satisfying points $(i)-(iv)$. \\

\noindent (\emph{If direction}): Suppose there exists a binary distribution $G^\dagger\in\mathcal{G}(F)$ satisfying points $(i)$-$(iv)$ of the lemma. Then $G$ solves \eqref{eq:4} because for all $\widehat G\in\mathcal{G}(F)$, 
\begin{align*}
    c_G(\underline r)-c_{G^\varnothing}(\underline r)+G(\underline r)\delta V &= c_{G^\dagger}(\underline r)-c_{G^\varnothing}(\underline r)+G^\dagger(\underline r)\delta V\geq c_{\widehat G}(\underline r)-c_{G^\varnothing}(\underline r)+\widehat G(\underline r)\delta V,
\end{align*}
where the equality holds by points $(ii)$ and $(iii)$ of the lemma while the inequality holds by point $(iv)$ of the lemma. 
\end{proof}\\

Next, we show that if a binary distribution $G$ solves \eqref{eq:4}, then $c_{G}$ must be tangent to $c_F$ at some point $x\in \text{int}(\Theta)$.

\begin{lemma}\label{lemma:interval-x}
There exists a unique point $\bar x\in(\underline r,\bar \theta)$ such that $\mathbb{E}_F[\theta|\theta\leq \bar x]= \underline r$, and $\mathbb{E}_F[\theta|\theta\leq x]< \underline r<\mathbb{E}_F[\theta|\theta>x]$ if and only if $x<\bar x$. Furthermore, if a binary distribution $G\in\mathcal{G}(F)$  is a solution to  \eqref{eq:4}, then $G=G^{\pi(x)}$ for some $x\in[\underline r,\bar x]$.
\end{lemma}

The lemma states two things: First--recalling that a distribution $G$ such that $c_G$ is tangent to $c_F$ at point $x$ is induced by pass/fail signal $\pi(x)$--the agent is persuaded to stop his search conditional on``pass" and to continue his search conditional on ``fail" if and only if  $x\leq \bar x$. Second, any $G^{\pi(x)}$ with $x<\underline r$ yields a strictly lower payoff than $G^{\pi(\underline r)}$, implying that no point of tangency strictly below $\underline r$ can solve \eqref{eq:4}. \\

\noindent\begin{proof}[Proof of \autoref{lemma:interval-x}]
For all $x\in \text{int}(\Theta)$, we have $\underline r=\delta\underline u<\underline u=\mathbb{E}_F[\theta] \leq \mathbb{E}_F[\theta|\theta>  x]$.  However, not all $x\in\text{int}(\Theta)$ satisfy $\mathbb{E}_F[\theta|\theta\leq x]\leq \underline r$. 
Since $F$ is assumed to be absolutely continuous, the function $x\mapsto \mathbb{E}_F[\theta|\theta\leq x]$ is continuous and strictly increasing. Additionally, $\mathbb{E}_F[\theta|\theta\leq \underline r]<\underline r<\mathbb{E}_F[\theta|\theta\leq \bar \theta]$. Hence, there is a unique point $\bar x\in(\underline r, \bar \theta)$ such that $\mathbb{E}_F[\theta|\theta\leq x]\leq \underline r$ if and only if $x\leq \bar x$, with strict inequality whenever $x<\bar x$. This proves the first statement of the lemma.

 To prove the second statement, suppose for the sake of a contradiction that a binary distribution $G\in\mathcal{G}(F)$ which solves \eqref{eq:4} has $c_G(x)<c_F(x)$ for all $x\in\text{int}(\Theta)$. From \autoref{lemma:binary}, $G$ has  support $\{m_1, m_2\}$ satisfying $\underline \theta<m_1\leq \underline r<m_2<\bar\theta$.\footnote{As $F$ is assumed to be absolutely continuous, no distribution $G\in \mathcal{G}(F)$ has positive mass at $\underline \theta$ or $\bar \theta$. Thus, $\underline\theta<m_1$ and $m_2<\bar \theta$.} Therefore, $G(\underline r)=G(m_1) \in (0,1).$ For $\epsilon >0$, consider a distribution $\widetilde G$ given by 
\[
\widetilde G(m)=\left\{\begin{array}{ccl}
     0 &\mbox{if} &~~~ m<m_1-(1-G(\underline r))\epsilon  \\
        G(\underline r) &\mbox{if} &~~~ m_1-(1-G(\underline r))\epsilon \leq m<m_2+G(\underline r)\epsilon  \\
     1 &\mbox{if} &~~~ m\geq m_2+G(\underline r)\epsilon 
     \end{array}\right..
\]
We make the following three observations. First, $m_1-(1-G(\underline r))\epsilon<m_1\leq \underline r< m_2\leq m_2+G(\underline r)\epsilon$, so $\widetilde G(\underline r)=G(\underline r)$. Second, by construction,  $c_G\leq c_{\widetilde G}$ pointwise, with $c_G(\underline r)<c_{\widetilde G}(\underline r)$. Third, notice that $\lVert c_G-c_{\widetilde G}\rVert_\infty=G(\underline r)(1-G(\underline r))\epsilon<\epsilon$. These three observations imply, respectively, that  $\widetilde G$ persuades the agent to stop with the same probability as does $G$, $\widetilde G$ yields a strictly higher price than does $G$, and for small enough $\epsilon$ that $\widetilde G$ is a mean-preserving contraction of $F$, i.e. $\widetilde G\in\mathcal{G}(F)$. Therefore, we have constructed, for small enough $\epsilon$, a distribution which contradicts the assumption that $G$ is a solution to \eqref{eq:4}.

Therefore, any binary distribution $G$ that solves \eqref{eq:4} must have $c_G$ tangent to $c_F$ at some  $x\in\text{int}(\Theta)$, i.e., there exists some $x\in\text{int}(\Theta)$ such that $G=G^{\pi(x)}$. It remains only to show that $x\in[\underline r,\bar x].$ Such a binary distribution $G^{\pi(x)}$ must have $ \mathbb{E}_F[\theta|\theta\leq x]\leq \underline r<\mathbb{E}_F[\theta|\theta>  x]$. Hence, by point $(i)$ of \autoref{lemma:binary}, $x\leq\bar x$. By the definition of $G^{\pi(\underline r)}$, we have that $c_{G^{\pi(\underline r)}}(\underline r)=c_F(\underline r)> c_{G^{\pi(x)}}(\underline r)$ for all $x\neq \underline r$. Additionally, for $x<\underline r$, we have $G^{\pi(x)}(\underline r)=F(x)<F(\underline r)=G^{\pi(\underline r)}(\underline r)$. In other words, for all $x<\underline r$, the value of \eqref{eq:4} evaluated at $G^{\pi(x)}$ is strictly lower than its value evaluated at $G^{\pi(\underline r)}$. Hence, $x\geq \underline r$. This proves the second statement of the lemma. 
\end{proof}\\

A point of tangency $x\in [\underline r, \bar x]$ pins down both the probability of persuading the agent to continue his search, $G^{\pi(x)}(\underline r)=F(x)$, and the price $c_{G^{\pi(x)}}(\underline r)-c_{G^\varnothing}(\underline r)$. The agent searches for the longest duration (largest slope) when $x=\bar x$, and he pays the highest per-period price when $x=\underline r$ because $c_{G^{\pi(\underline r)}}(\underline r)\geq c_{G^{\pi(x)}}(\underline r)$ for all $x$ as previously mentioned.  We can further rewrite the price  associated with distribution $G^{\pi(x)}$ as
 \begin{align*}\label{eq:price}
     c_{G^{\pi(x)}}(\underline r)-c_{G^\varnothing}(\underline r)&=c_{G^{\pi(x)}}(r(G^{\pi(x)}))+\int^{r(G^{\pi(x)})}_{\underline r}\big(1-G^{\pi(x)}(m)\big)dm-c_{G^\varnothing}(\underline r)\\[8pt]
     &=c_{G^{\pi(x)}}(r(G^{\pi(x)}))+\big(r(G^{\pi(x)})-\underline r\big)\big(1-F(x)\big)-c_{G^\varnothing}(\underline r) \notag\\[8pt]
     &=\big(u(G^{\pi(x)})-\underline u\big)\big(1-\delta F(x)\big) \tag{5},
 \end{align*}
 where $r(G^{\pi(x)})=\delta u(G^{\pi(x)})$ is the solution to \eqref{eq:2}, that is, the agent's reservation value if he were to observe $G^{\pi(x)}$ for free in every period. The first equality above follows from breaking up the integral in the definition of $c_{G^{\pi(x)}}(\underline r)$, the second equality follows by definition of $G^{\pi(x)}$ and the fact that $\mathbb{E}_F[\theta|\theta\leq x]\leq \underline r\leq  r(G^{\pi(x)})<\mathbb{E}_F[\theta|\theta>x]$,\footnote{For all $x\in [\underline r, \bar x]$, we have  $\mathbb{E}_F[\theta|\theta\leq x]\leq r(G^{\pi(x)})<\mathbb{E}_F[\theta|\theta>x]$. The first inequality follows because $\underline r\leq r(G^{\pi(x)})$ and  $\mathbb{E}_F[\theta|\theta\leq x]\leq \underline r$. The second inequality follows because if $\mathbb{E}_F[\theta|\theta>x]\leq r(G^{\pi(x)})$, then $c_{G^{\pi(x)}}(r(G^{\pi(x)}))=0$, which by \eqref{eq:2} would imply that $r(G^{\pi(x)})=0<\underline r$, yielding a contradiction.} and the last equality follows from evaluating the fixed point problems in \eqref{eq:2} at $G^{\pi(x)}$ and $G^\varnothing$.  Hence, the value of the objective function in \eqref{eq:4} is entirely determined by the point of tangency, which reduces our infinite dimensional optimization problem to a single dimensional problem:
 \begin{align*}
&\max_{x\in[\underline r, \bar x]} \hspace*{.2em} \big(u(G^{\pi(x)})-\underline u\big)\big(1-\delta F(x)\big)+\delta F(x)V.
\end{align*}

Of course, in a stationary equilibrium, the principal's continuation value $V$ must be consistent with \eqref{eq:sgp}, i.e., 
\begin{align*}
V=&\max_{x\in[\underline r, \bar x]} \hspace*{.2em} \big(u(G^{\pi(x)})-\underline u\big)\big(1-\delta F(x)\big)+\delta F(x)V\\[8pt]
\label{eq:404}
\tag{6}
=&\max_{x\in[\underline r, \bar x]} \hspace*{.2em} u(G^{\pi(x)})-\underline u.
\end{align*}
Thus, the principal's maximization problem reduces to a surplus maximization problem. The following lemma demonstrates that $x=\bar r$ is the unique point of tangency that solves \eqref{eq:404}.

\begin{lemma}\label{lemma:surplus_max}
 The McCall reservation value $\bar r$ lies in the open interval $(\underline r, \bar x)$, and is the unique solution to \eqref{eq:404}.
\end{lemma}

\noindent\begin{proof}[Proof of \autoref{lemma:surplus_max}] 
We first prove that $\bar r\in (\underline r, \bar x)$. By  \hyperref[ass:1]{\Cref{ass:1}}, we have $\underline r<\bar r$. Hence, we need only prove $\bar r<\bar x$. From  \autoref{lemma:interval-x}, proving that $\bar r< \bar x$  is equivalent to showing that
$\mathbb{E}_F[\theta|\theta\leq \bar r]<\underline r$. Suppose for contradiction that $\mathbb{E}_F[\theta|\theta\leq \bar r]\geq\underline r$. Then
\begin{align*}
\bar r&=\left(\frac{\delta}{1-\delta}\right)\cdot c_{G^{\pi(\bar r)}}(\bar r)\\[6pt]
&\leq \left(\frac{\delta}{1-\delta}\right)\cdot c_{G^{\pi(\bar r)}}(\mathbb{E}_F[\theta|\theta\leq \bar r])\\[6pt]
&=\left(\frac{\delta}{1-\delta}\right)\cdot c_{G^\varnothing}(\mathbb{E}_F[\theta|\theta\leq \bar r])\\[6pt]
&\leq \left(\frac{\delta}{1-\delta}\right)\cdot c_{G^\varnothing}(\underline r)\\[6pt]
&=\underline r,
\end{align*}
where the first equality follows because $\bar r=r(G^{\pi(\bar r)})$,  the first inequality follows because $c_{G^{\pi(\bar r)}}(\cdot)$ is weakly decreasing and $\mathbb{E}_F[\theta|\theta\leq \bar r]<\bar r$, the second equality follows because $c_{G^{\pi(\bar r)}}(x')=c_{G^\varnothing}(x')$ for all  $x'\notin \big(\mathbb{E}_F[\theta|\theta\leq \bar r], \mathbb{E}_F[\theta|\theta>\bar r]\big)$ by \eqref{eq:3}, the last inequality follows because $c_{G^\varnothing}(\cdot)$ is weakly decreasing and $\underline r\leq \mathbb{E}_F[\theta|\theta\leq \bar r]$ from our initial assumption for the sake of a contradiction, and the last equality follows because $\underline r$ is the unique solution to \eqref{eq:2} evaluated at $G^\varnothing$. However, this yields $\bar r\leq \underline r$, which is not possible given \hyperref[ass:1]{\Cref{ass:1}}. Hence, $\mathbb{E}_F[\theta|\theta\leq \bar r]<\underline r$, which proves that $\bar r<\bar x$.

Next, we prove that $\bar r$ is the unique solution to \eqref{eq:404}. For any $x\in\text{int}(\Theta)\backslash\{\bar r\}$,
\[
\bar r=\left(\frac{\delta}{1-\delta}\right)\cdot c_{F}(\bar r)=\left(\frac{\delta}{1-\delta}\right)\cdot c_{G^{\pi(\bar r)}}(\bar r)>\left(\frac{\delta}{1-\delta}\right)\cdot c_{G^{\pi(x)}}(\bar r)
\]
where the first equality follows because $\bar r$ is the unique solution to \eqref{eq:2} evaluated at $F$, the second equality follows because $c_{G^{\pi(\bar r)}}(\cdot)$ is tangent to $c_F$ at $\bar r$, and the strict inequality follows because $c_{G^{\pi(x)}}(x')=c_F(x')$ only for $x'\in \{\underline \theta, x, \bar \theta\}$. Since  $r(G^{\pi(x)})$  is the unique solution to \eqref{eq:2} evaluated at $G^{\pi(x)}$,  $\bar r\neq r(G^{\pi(x)})$. In fact, since $\bar r$ is the highest possible reservation value, we conclude that $\bar r> r(G^{\pi(x)})$. Finally, because a reservation value is equal to the agent's discounted continuation value, we have $u(G^{\pi(x)})<\bar u$ for any $x\in\text{int}(\Theta)\backslash\{\bar r\}$. \end{proof}\\

 Thus, the unique solution to \eqref{eq:404} is at $x=\bar r$, which yields $V=\bar u-\underline u$, establishing point $(ii)$ of the theorem. Furthermore, it implies that   $G^{\pi(\bar r)}$ is the unique binary distribution that solves \eqref{eq:4}.
 
 So far, we have shown that if the tuple $(p, G, U, V)$ constitutes a stationary equilibrium, then the agent's continuation value is $U=\underline u$ and the principal's continuation value is $V=\bar u-\underline u$. From  \eqref{eq:price}, we pin down the price as $p=(\bar u-\underline u)(1-\delta F(\bar r))$, which establishes point $(iii)$ of the theorem. Furthermore, by point $(ii)$ of \autoref{lemma:binary}, $G(\underline r)=G^{\pi(\bar r)}(\underline r)$, and the latter is equivalent to $F(\bar r)$ by construction, establishing point $(iv)$ of the theorem. Finally, by point $(iii)$ of \autoref{lemma:binary}, $c_{G^{\pi(\bar r)}}(x)\leq c_G(x)$ for all $x\in\Theta$ with equality at $x=\underline r$, establishing point $(v)$ of the theorem. 

\subsection*{If direction}

 \noindent We now complete the proof by showing that points $(i)$-$(v)$ are sufficient conditions for a stationary equilibrium. To that end, suppose the principal proposes a contract $\langle p, G\rangle$ that satisfies points $(iii)$-$(v)$  of \autoref{prop:1} in every period following any history of events. Additionally, suppose the agent and the principal anticipate  continuation values satisfying $(i)$ and $(ii)$ of 
\autoref{prop:1}, i.e., $U=\underline u$ and $V=\bar u-\underline u$, respectively, following any history.
 We must show that \eqref{eq:os}, \eqref{eq:pc}, \eqref{eq:pm}, \eqref{eq:sga}, and \eqref{eq:sgp} are all satisfied.

As the agent anticipates the same continuation value $U=\underline u$ following any history of events, sequential rationality implies that he should stop searching in any given period if he samples a good whose expected quality satisfies $m>\delta \underline u=\underline r$. Hence, \eqref{eq:os} is satisfied. 

If the agent accepts the contract $\langle p, G\rangle$, he gets an expected utility of 
\begin{align*}
      \underline r+c_G( \underline r)-p &= \underline r+c_{G^{\pi(\bar r)}}(\underline r)-(\bar u-\underline u)\big(1-\delta F(\bar r)\big)\\[6pt]
     &= \underline r+c_{G^{\pi(\bar r)}}(\bar r)+\int^{\bar r}_{\underline r}\big(1-G^{\pi(\bar r)}(m)\big)dm-(\bar u-\underline u)\big(1-\delta F(\bar r)\big)\\[6pt]
     &= \underline r+c_{F}(\bar r)+(\bar r-\underline r)(1-F(\bar r))-(\bar u-\underline u)\big(1-\delta F(\bar r)\big)\\[6pt]
    &= \underline r+\left(\frac{1-\delta}{\delta}\right)\bar r+(\bar r-\underline r)(1-F(\bar r))-(\bar u-\underline u)\big(1-\delta F(\bar r)\big)\\[6pt]
    &=\underline r+\left(\frac{1-\delta}{\delta}\right)\underline r\\[6pt]
    &=\underline r+c_{G^\varnothing}(\underline r),
     \end{align*}
where the first equality follows by points $(iii)$ and $(v)$
 of the theorem, the second equality follows by breaking up the integral in $c_{G^{\pi(\bar r)}}(\underline r)$, the third equality follows because  $c_{G^{\pi(\bar r)}}$ is tangent to $c_F$ at $x=\bar r$ and because $G^{\pi(\bar r)}(m)=F(\bar r)$ for all $m\in \big[\mathbb E_F[\theta|\theta\leq \bar r], \mathbb E_F[\theta|\theta>\bar r]\big]\supset [\underline r, \bar r]$, the fourth equality follows because $\bar r$ is the unique reservation value that solves \eqref{eq:2} evaluated at $F$, and the last equality follows because $\underline r$ is the unique reservation value that solves \eqref{eq:2} evaluated at $G^\varnothing$. On the other hand, if the agent rejects the contract, he gets an expected utility of $\underline r+c_{G^\varnothing}(\underline r)$. Hence, the agent's participation constraint \eqref{eq:pc} binds for the contract $\langle p, G\rangle$. 
 
Given that the agent will accept the contract, his expected utility in any given period is 
\begin{align*}
\delta \underline u+c_G(\delta\underline u)-p=&\delta\underline u+c_{G^\varnothing}(\delta\underline u)\\
=&\underline u,
\end{align*}
where the first equality follows because the constraint \eqref{eq:pc} binds, and the second equality follows from \eqref{eq:1'}. Hence, the agent's self-generation condition \eqref{eq:sga} is satisfied.

Similarly, the principal's expected utility in any given period is
\begin{align*}
p+\delta G(\underline r)(\bar u-\underline u)=&(\bar u-\underline u)(1-\delta F(\bar r)) +\delta F(\bar r)(\bar u-\underline u)\\
=& \bar u-\underline u,
\end{align*}
where the first equality follows from $(iii)$ and $(iv)$ of \autoref{prop:1}. Hence, the principal's self-generation condition \eqref{eq:sgp} is also satisfied. Furthermore, since $\bar u-\underline u$ is the maximum payoff the principal can ever earn, \eqref{eq:pm} is also satisfied. Hence, a pair of continuation values $(U,V)=(\underline u, \bar u-\underline u)$ along with any contract $\langle p, G\rangle$ satisfying $(iii)$-$(v)$ of \autoref{prop:1} constitutes a stationary equilibrium. Clearly, the contract  $\langle p, G\rangle=\langle (\bar u-\underline u)\big(1-\delta F( \bar r)\big),G^{\pi(\bar r)}\rangle$ is one such contract, proving existence.
\end{proof}\\

\noindent\begin{proof}[Proof of \autoref{cor:1}]
From point $(iv)$ of \autoref{prop:1}, any  $G\in\mathcal{G}(F)$ that supports a stationary equilibrium must satisfy $G( \underline r)=F(\bar r)$, which implies by point $(d)$ of \autoref{remark:1} that $\partial_+ c_G( \underline r)=F(\bar r)-1$. Additionally, from point $(v)$ of \autoref{prop:1}, $c_G\geq c_{G^{\pi(\bar r)}}$  
pointwise. We also have $c_F\geq c_G$ pointwise for any $G\in\mathcal{G}(F)$. By construction, $c_F(\bar r)=c_{G^{\pi(\bar r)}}(\bar r)$, and $\partial_+ c_F( \bar r)=\partial_+ c_{G^{\pi(\bar r)}}( \bar r)=F(\bar r)-1$. Since $c_G$ is sandwiched between $c_F$ and $c_{G^{\bar r}}$, we have $\partial_+ c_G( \bar r)=F(\bar r)-1$ as well. 
Finally, by the convexity of $c_G$ (see point $(c)$ of \autoref{remark:1}), we have 
\[
\underbrace{\partial_+ c_G(\underline r)}_{=F(\bar r)-1}\hspace*{.5em}\leq \hspace*{.5em}\underbrace{\partial_+ c_G(x)}_{=G(x)-1}\hspace*{.5em}\leq \hspace*{.5em}\underbrace{\partial_+ c_G(\bar r)}_{=F(\bar r)-1}
\]for all $x\in [\underline r, \bar r]$. Equivalently, $G(x)=F(\bar r)$ for all $x\in [\underline r,  \bar r]$.
\end{proof}\\

\noindent\begin{proof}[Proof of \autoref{cor_delta_comp}]
Fix $\Theta$ and an absolutely continuous distribution $F$ over $\Theta$. For the purposes of this proof, we consider a search environment indexed by the discount factor $\mathcal{M}^\delta\triangleq(\Theta,F,\delta)$. Define the mapping $\delta \mapsto \bar r(\delta)$ to be the McCall reservation value for environment $\mathcal M^\delta$, i.e., $\bar r(\delta)$ is the unique solution of \eqref{eq:2} for $G=F$ when the discount factor is $\delta$. Note that $\bar r(\cdot)$ is a continuous and increasing function. Furthermore, $\bar r(\cdot)$ is strictly increasing over the open interval  $(\max\{0, \underline \theta/m^\varnothing\}, 1)$, the region over which \hyperref[ass:1]{\Cref{ass:1}} is satisfied. Finally, $\bar r(\delta)\to \bar \theta$ as $\delta\to 1$, i.e., in a McCall setting, the agent searches until he finds the highest quality good when search becomes costless.

Let $\delta'$ and $\delta''$ be such that $0<\delta'<\delta''<1$. We show that the probability with which the agent terminates his search in any given period is weakly higher in the essentially unique stationary equilibrium outcome of $\mathcal{M}^{\delta'}$ than that of $\mathcal{M}^{\delta''}$, which therefore implies that he searches in expectation for a shorter duration in $\mathcal{M}^{\delta'}$ than he does in $\mathcal{M}^{\delta''}$. To that end, first note that if \hyperref[ass:1]{\Cref{ass:1}} is not satisfied for $\mathcal{M}^{\delta'}$, then the corollary trivially holds because the agent terminates his search with probability one in the first period (see \autoref{lemma:1}). Next, suppose \hyperref[ass:1]{\Cref{ass:1}} is satisfied for $\mathcal{M}^{\delta'}$, which implies it is also satisfied for $\mathcal{M}^{\delta''}$. From \autoref{prop:1}, given $\mathcal M^\delta$, the the agent terminates his search in any period with probability $1-F(\bar r(\delta))$. The statements of the corollary then immediately follow from the properties of $\bar r(\cdot)$ we outlined above, namely, $\bar r(\delta)$ is strictly increasing in $\delta$ and converges to $\bar \theta$ as $\delta\to 1$. \end{proof}\\

\noindent\begin{proof}[Proof of \autoref{prop:2}]
     Suppose a family of contracts $\{\langle p_z,G_z\rangle\}_{z\in Z}$ 
 along with a pair of agent-principal continuation values $(U, V)$ constitute a stationary equilibrium. Since the objective function in \eqref{eq:pm2} is increasing in the price $\widehat p$, the constraint \eqref{eq:pc2} must   bind at the maximum, i.e., $p_z=c_{G_z}(\delta U)-c_{G_z^\varnothing}(\delta U)$ for all $z\in Z$. The fixed-point problem in \eqref{eq:sga2} then becomes  
\begin{align*}
U=&\delta U+\int_Z c_{G_z^\varnothing}(\delta U)\xi(dz)\\[6pt]
=&\delta U+c_{G^\xi}(\delta U),
\end{align*}
where the second equality holds because the mapping $G\to c_G$ is linear, which implies $\mathbb{E}_\xi[c_{G^\varnothing_z}(\cdot)]=c_{\mathbb{E}_\xi[G^\varnothing_z]}(\cdot)$, and by law of total probability, $\mathbb{E}_\xi[G^\varnothing_z]=G^\xi$. 
The above equality is the same as \eqref{eq:1'} evaluated at $G^\xi$ and thus has a unique solution $U=u(G^\xi)$, establishing point $(i)$ of the theorem. The agent's optimal stopping rule is then characterized by the reservation value $\delta u(G^\xi)\triangleq r^\xi$ from \eqref{eq:os}.

In any equilibrium, the principal can always guarantee herself a payoff of zero and cannot extract more than the maximal surplus. Hence, the principal's stationary equilibrium continuation value satisfies $V\in [0, \bar u-u(G^\xi)]$. To prove point $(ii)$ of the theorem, we proceed in two steps: First, given the principal's  continuation payoff $V$, we characterize the principal's profit according to \eqref{eq:pm2} following any realization $z\in Z$ of the public signal. Second, we show that there is a unique $V$ that is consistent with the principal's self-generation condition \eqref{eq:sgp2}. \\

\noindent\textbf{Step 1:} Given the agent's reservation value and the binding constraint \eqref{eq:pc2}, the principal's profit maximization in \eqref{eq:pm2} following any $z\in Z$ can be expressed as 
\[
\label{eq:7}
\tag{7}
  G_z\in\argmax_{\widehat G\in \mathcal{G}(F_z)}\hspace*{.2em} c_{\widehat G}(r^\xi)-c_{G^\varnothing_z}(r^\xi)+ \widehat G(r^\xi) \delta V.
\]   
Let $\psi_z(V)$ represent the value function of the maximization problem associated with \eqref{eq:7}. Below, we fully characterize the principal's value function for any public signal realization. To that end, we fix a public signal realization $z\in Z$ and consider three exhaustive cases depending on the support of the interim belief $F_z$:\\

\noindent \textbf{Case 1: $r^\xi<\underline \theta_z$.} In this case, the agent's interim belief is too optimistic as even the lowest quality possible given his beliefs will fall above his reservation value. Hence, he has no value for information and will stop his search with probability one. Formally, $c_{F_z}(x)=c_{G^\varnothing_z}(x)$ for all $x< \underline \theta_z$. Thus, for all $\widehat G\in\mathcal{G}(F_z)$, $c_{\widehat G}(r^\xi)=c_{G^\varnothing_z}(r^\xi)$ and $\widehat G(r^\xi)=G^\varnothing_z(r^\xi)=0$. Thus, $\psi_z(V)=0$ and the principal can achieve this value by offering contract $\langle p_z, G_z\rangle=\langle 0, G^\varnothing_z\rangle$. \\

\noindent \textbf{Case 2: $r^\xi>\bar \theta_z$.} In this case, the agent's interim belief is too pessimistic as even the highest quality possible given his beliefs will fall short of his reservation value. Hence, he has no value for information and will continue to search with probability one. 

Formally,  for all $\widehat G\in \mathcal{G}(F_z)$, $0=c_{F_z}(\underline \theta_z)\geq c_{F_z}(r^\xi)\geq c_{\widehat G}(r^\xi)\geq c_{G^\varnothing_z}(r^\xi)=0$, where the first inequality follows because $c_{F_z}$ is a weakly decreasing function, and the second and third inequalities follows because $c_{F_z}\geq c_{\widehat G}\geq c_{G^\varnothing_z}$ pointwise, and the last equality follows because $r^\xi>\underline\theta_z$. Hence, $c_{\widehat G}(r^\xi)=0$ which further implies $\widehat G(r^\xi)=1$. Thus, $\psi_z(V)=\delta V$ and the principal can achieve this value by offering contract $\langle p_z, G_z\rangle=\langle 0, G^\varnothing_z\rangle$.\\

\noindent \textbf{Case 3: $r^\xi\in [\underline \theta_z, \bar \theta_z]$.} 
In this case, the agent has a positive WTP for information, and thus is willing to contract with the principal. We next show that without loss of optimality, we can restrict attention to a sub-class of pass/fail signals $\{\pi(x_z)\}_{x_z\in [\underline \theta_z, \bar \theta_z]}$. 

\begin{lemma}\label{lemma_theorem_2_proof}
Suppose $r^\xi\in [\underline \theta_z, \bar \theta_z]$. Then the principal's value function can be attained by a pass/fail signal $\pi(x_z)$ for some $x_z\in [r^\xi, \bar x_z]$.
\end{lemma}
\begin{proof}[Proof of  \autoref{lemma_theorem_2_proof}]
The proof, which follows a similar argument to that of \autoref{lemma:interval-x}, is omitted.
\end{proof}\\
 
We can now express the profit maximization problem in \eqref{eq:7} as solving for the optimal pass/fail cutoff:
\begin{align*}
  x_z\in\argmax_{x\in [r^\xi, \bar x_z]}\hspace*{.2em} c_{F_z}(x)+\big(1-F_z(x)\big)(x-r^\xi)-c_{G^\varnothing_z}(r^\xi)+ F_z(x) \delta V.
\end{align*}  
 This objective function is (weakly) increasing for  $x<\delta V+r^\xi$ and (weakly) decreasing for $x>\delta V+r^\xi$. Thus, $x_z=\min\{\bar x_z, \delta V+r^\xi\}$ is a solution, and the principal's value function is 
\[
\psi_z(V)=\left\{\begin{array}{ccl}
 c_{F_z}(\delta V+r^\xi)-c_{G^\varnothing_z}(r^\xi)+\delta V &\mbox{if} & \delta V+r^\xi\leq \bar x_{z}\\
 F_z(\bar x_z)\delta V   & \mbox{if} &   \delta V+r^\xi> \bar x_{z}\end{array}\right..
\]
That $\psi_z(V)=  c_{F_z}(\delta V+r^\xi)-c_{G^\varnothing_z}(r^\xi)+\delta V$ if $\delta V+r^\xi\leq \bar x_{z}$ follows directly from substituting $\delta V+r^\xi$ into the objective function for the optimal pass/fail cutoff. That $\psi_z(V)= F_z(\bar x_z)\delta V$ if $\delta V+r^\xi> \bar x_{z}$ follows from the fact that
\begin{align*}
c_{F}(\bar x_z)+\big(1-F_z(\bar x_z)\big)(\bar x_z-r^\xi)&=\int^{\bar \theta_z}_{\bar x_z}(m-r^\xi)dF_z(m)\\[6pt]
&=\mathbb{E}_{F_z}[\theta]-r^\xi-\Big(\mathbb{E}_{F_z}[\theta|\theta\leq \bar x_z]-r^\xi\Big)F_z(\bar x_z)\\[6pt]
&=\max\{\mathbb{E}_{F_z}[\theta]-r^\xi, 0\}\\[6pt]
\label{eq:9}
\tag{8}
&= c_{G^\varnothing_z}(r^\xi)
\end{align*}
where the third equality follows because either $\mathbb{E}_{F_z}[\theta]\leq r^\xi$, in which case $\bar x_z=\bar\theta_z$ implying that $\mathbb{E}_{F_z}[\theta|\theta\leq \bar x_z]=\mathbb{E}_{F_z}[\theta]$ and $F_z(\bar x_z)=1$, or $\mathbb{E}_{F_z}[\theta]> r^\xi$, in which case $\bar x_z<\bar\theta_z$ and $\mathbb{E}_{F_z}[\theta|\theta\leq \bar x_z]=r^\xi$. This concludes Case 3.

We have now characterized the value function for all possible cases, which concludes Step 1. For convenience, we summarize below the principal's value function:
\small{\[
\psi_z(V)=\left\{\begin{array}{ccl}
  0&\mbox{if} & z\in Z_1\triangleq \{z'\in Z: r^\xi<\underline \theta_{z'}\}\\
\delta V&\mbox{if} &z\in Z_2\triangleq \{z' \in Z: r^\xi> \bar \theta_{z'}\}\\
    c_{F_z}(\delta V+r^\xi)-c_{G^\varnothing_z}(r^\xi)+\delta V &\mbox{if} &  z\in Z_3(V)\triangleq\{z'\in Z: r^\xi\in[\underline\theta_{z'}, \bar\theta_{z'}] \hspace*{.1em} \wedge \hspace*{.1em}  \bar x_{z'}\geq \delta V+r^\xi\} \\
  F_z(\bar x_z)\delta V   & \mbox{if} &  z\in Z_4(V)\triangleq\{z'\in Z: r^\xi\in[\underline\theta_{z'}, \bar\theta_{z'}] \hspace*{.1em} \wedge \hspace*{.1em}  \bar x_{z'}<\delta V+r^\xi\}
\end{array}\right..
\]}
\normalsize

\noindent \textbf{Step 2:} The principal's continuation value $V$ must be consistent with the principal's self-generation condition \eqref{eq:sgp2}, i.e., 
\small{\begin{align*}
 V=&\int_Z \psi_z(V)\xi(dz)\\[6pt]
    =& \delta V \xi(Z_2)+\int_{Z_3(V)}\left(c_{F_z}(\delta V+r^\xi)-c_{G^\varnothing_z}(r^\xi)+\delta V\right) \xi(dz) + \int_{Z_4(V)}\delta VF_z(\bar x_z) \xi(dz)\\[6pt]
    =&\int_{Z}\left(c_{F_z}(\delta V+r^\xi)-c_{G^\varnothing_z}(r^\xi)+\delta V \right)\xi(dz) -\int_{Z_1}\left(c_{F_z}(\delta V+r^\xi)-c_{G^\varnothing_z}(r^\xi)+\delta V \right)\xi(dz)\\[6pt]
    &-\int_{Z_2}\left(c_{F_z}(\delta V+r^\xi)-c_{G^\varnothing_z}(r^\xi)\right) \xi(dz)-\int_{Z_4(V)}\left(c_{F_z}(\delta V+r^\xi)-c_{G^\varnothing_z}(r^\xi)+\big(1-F_z(\bar x_z)\big)\delta V \right)\xi(dz)  \label{eq:10}
    \tag{9}.
\end{align*}}
\normalsize
Let us simplify each of these terms separately. The first term in \eqref{eq:10} simplifies to
\begin{align*}
    \int_{Z}\left(c_{F_z}(\delta V+r^\xi)-c_{G^\varnothing_z}(r^\xi)+\delta V\right) \xi(dz)&=c_{F}(\delta V+r^\xi)-c_{G^\xi}(r^\xi)+\delta V \\[6pt]
    &=c_{F}(\bar r)+\int_{\delta V+r^\xi}^{\bar r}\big(1-F(m)\big)dm-c_{G^\xi}(r^\xi)+\delta V\\[6pt]
    &=(1-\delta)\big(\bar u-u(G^\xi)\big)+\int_{\delta V+r^\xi}^{\bar r}\big(1-F(m)\big)dm +\delta V 
\end{align*}
where the first equality follows because the mapping $G\to c_G$ is linear and because $\mathbb{E}_\xi[F_z]=F$ and $\mathbb{E}_\xi[G^\varnothing_z]=G^\xi$, the second follows by breaking up the integral in the definition of $c_F(r^\xi+\delta V)$, and the last equality follows by evaluating \eqref{eq:1'} at $F$ and $G^\xi$.

For all $z\in Z_1$, we have $c_{F_z}(r^\xi)=c_{G^\varnothing_z}(r^\xi)$. Thus, the integrand of the second term in \eqref{eq:10} simplifies to 
\begin{align*}
    c_{F_z}(\delta V+r^\xi)-c_{G^\varnothing_z}(r^\xi)+\delta V &=c_{F_z}(r^\xi)-\int^{\delta V+r^\xi}_{r^\xi}\hspace*{-.25em}\big(1-F_z(m)\big)dm-c_{G^\varnothing_z}(r^\xi)+\delta V\\[6pt]
    &=\int^{\delta V+r^\xi}_{r^\xi}\hspace*{-.5em}F_z(m)dm.
\end{align*}
Note that the above integral is non-zero only if $\delta V+r^\xi> \underline\theta_z$, which is equivalent to $z\in Z_A(\bar r-r^\xi-\delta V)$. 

For all $z\in Z_2$, we have $\delta V+r^\xi\geq r^\xi>\bar\theta_z$, which implies that $c_{F_z}(\delta V+r^\xi)=c_{G^\varnothing_z}(r^\xi)=0$. Thus, the third term in \eqref{eq:10} is zero.

Finally, for all $z\in Z_4(V)$, the integrand in the last term simplifies to 
\small{\begin{align*}
    c_{F_z}(\delta V+r^\xi)-c_{G^\varnothing_z}(r^\xi)+\big(1-F_z(\bar x_z)\big)\delta V &=c_{F_z}(\bar x_z)-\int^{\delta V+r^\xi}_{\bar x_z}\big(1-F_z(m)\big)dm-c_{G^\varnothing_z}(r^\xi)+\big(1-F_z(\bar x_z)\big)\delta V\\[6pt]
    &=\int^{\delta V+r^\xi}_{\bar x_z}\left(F_z(m)-F_z(\bar x_z)\right)dm,
\end{align*}}%
\normalsize
where the first equality follows by breaking up the integral in the definition of $c_{F_z}(\delta V+r^\xi)$ and the second equality follows from \eqref{eq:9}. Note that the above integral is non-zero only if $\bar\theta_z>\bar x_z$, which is equivalent to $z\in Z_B(\bar r-r^\xi-\delta V)$.

Put together, we can write the expression in \eqref{eq:10} as 
\small{\begin{align*}
    V= &\bar u-u(G^\xi)+\frac{1}{1-\delta}\int_{\delta V+r^\xi}^{\bar r}1-F(m)dm\\[6pt]
    &-\frac{1}{1-\delta}\left(\int_{Z_A}\int^{\delta V+r^\xi}_{\underline\theta_z}F_z(m)dm \xi(dz)+\int_{Z_B}\int^{\delta V+r^\xi}_{\bar x_z}\left(F_z(m)-F_z(\bar x_z)\right)dm \xi(dz)\right)\\[6pt]
    =&\bar u-u(G^\xi)+\frac{k}{1-\delta}-\frac{1}{1-\delta}\Phi(k^*),
\end{align*}}%
\normalsize
where $\Phi(\cdot)$ is as given in \eqref{eq:5}. The last equality follows by  using a change of variable $V=\bar u-u(G^\xi)-k^*/\delta$ so that $\delta V+r^\xi=\bar r-k^*$. Since we have both $V=\bar u-u(G^\xi)+k^*/(1-\delta)-\Phi(k^*)/(1-\delta)$ and $V=\bar u-u(G^\xi)-k^*/\delta$, it must be the case that $k^*$ solves the fixed point 
\[
k=\delta\Phi(k),
\]
which gives us the representation in \autoref{prop:2}. 

It remains to show such a solution exists and is unique. Note $k-\delta\Phi(k)$ is continuous and strictly increasing in $k$ because both $Z_A(\cdot)$ and $Z_B(\cdot)$ are decreasing by construction (in the inclusion order). Additionally, 
\[
0-\delta\Phi(0)=-\delta\left(\int_{Z_A(k)}\int_{\underline \theta_z}^{\bar r}F_z(m)dm \hspace*{.2em}\xi(dz)+\int_{Z_B(0)}\int_{\bar x_z}^{\bar r}\left(F_z(m)-F_z(\bar x_z)\right)dm \hspace*{.2em}\xi(dz)\right)\leq 0.
\]
The first integral is non-negative because $\bar r>\underline\theta_z$ for all $z\in Z_A(0)$, and the second integral is non-negative because $\bar r>\bar x_z$ for all  $z\in Z_B(0)$. On the other hand, 
\[
\bar r-r^\xi-\delta\Phi(\bar r-r^\xi)=\int_{r^\xi}^{\bar r}
\big(1-\delta F(m)\big)dm>0
\] 
since $Z_A(\bar r-r^\xi)=Z_B(\bar r-r^\xi)=\emptyset$.
Thus, there exists a unique $k^*\in [0, \bar r-r^\xi]$ that solves the fixed point problem, which establishes point $(ii)$ of \autoref{prop:2} and completes Step 2.

We have found a unique pair $(U,V)$ and constructed a family of pass/fail signals that attains these values for the agent and principal, respectively, such that \autoref{def:2} is satisfied. Therefore, an equilibrium exists, completing the proof.
\end{proof}\\

\noindent \begin{proof}[Proof of \autoref{cor_info_comp}]
Fix $\Theta$, an absolutely continuous distribution $F$ over $\Theta$, and a discount factor $\delta$. Let $\mathcal M\triangleq (\Theta, F, \delta)$ be a search environment without a public signal, and let $\mathcal M^\xi\triangleq (\Theta, F, \delta, \xi, Z)$ be a search environment with public signal $(\xi, Z)$.  We show that the probability with which the agent terminates his search in any given period is weakly higher in the unique stationary equilibrium outcome of $\mathcal{M}^\xi$ than that of $\mathcal{M}$, which therefore implies that he searches in expectation for a shorter duration in $\mathcal{M}^\xi$ than he does in $\mathcal{M}$.

Recall that, from \autoref{prop:1}, the per-period probability that the agent terminates his search in the unique stationary equilibrium outcome of $\mathcal M$ is $1-F(\bar r)$ . Given $\mathcal M^\xi$, for each realization $z\in Z$ and any period $t\geq 1$, there are three exhaustive cases to consider: $(i)$ if $z\in Z_A(k^*)$, the agent stops his search with probability $1> 1-F_z(\bar r-k^*)$, where the inequality follows since $\bar r-k^*>\underline \theta_z$ by construction of $Z_A(k^*)$; $(ii)$ if $z\in Z_B(k^*)$, the agent stops with probability $1-F_z(\bar x_z)> 1-F_z(\bar r-k^*)$, where the inequality follows because $\bar r-k^*>\bar x_z$;\footnote{If $F_z(\bar x_z)= F_z(\bar r-k^*)$, then $\bar x_z=\bar r-k^*$ by definition of $\bar x_z$.} $(iii)$ if  $z\in\left(Z_A(k^*)\cup Z_B(k^*)\right)^\complement$, the agent stops  with probability $1-F_z(\bar r-k^*)$. Therefore, the per-period probability  that the agent terminates his search in the unique stationary equilibrium outcome of  $\mathcal M^\xi$ is given by 

\begin{align*}
\xi\left(Z_A(k^*)\right)+\hspace*{-.5em} \int\limits_{Z_B(k^*)}\hspace*{-.65em}\big(1-F_z(\bar x_z)\big)\xi(dz)+\hspace*{-2.25em}\int\limits_{\left(Z_A(k^*)\cup Z_B(k^*)\right)^\complement}\hspace*{-2.5em}\big(1-F_z(\bar r-k^*)\big)\xi(dz)& \geq \int\limits_{Z}\big(1-F_z(\bar r-k^*)\big)\xi(dz)\\[6pt]
&=1-F(\bar r-k^*)\\[6pt]
&\geq 1-F(\bar r),
\end{align*}
where the equality follows by Bayes-plausibility, and the last inequality follows because $k^*\geq 0$. This proves the first statement of the corollary.

Additionally, if $k^*>0$ then the last inequality is strict because  $F$ has full support, implying that the expected duration of search is strictly shorter in $\mathcal M^\xi$ than $\mathcal M$. Conversely, if $k^*=0$, which is possible only if $\xi(Z_A(k^*)\cup Z_B(k^*))=0$, the above two inequalities hold with equality, implying that the expected duration of search in $\mathcal M^\xi$ and $\mathcal M$ are equal. This proves the second statement of the corollary.
\end{proof}

  \endgroup

\end{document}